\documentclass[conference,letterpaper]{IEEEtran}
\usepackage[utf8]{inputenc}
\usepackage[T1]{fontenc}
\usepackage{ifthen}
\usepackage{cite}
\usepackage[cmex10]{amsmath} 
\usepackage{amssymb} 
\usepackage{mathrsfs}
\usepackage{amsthm}
\usepackage{graphicx}
\newtheorem{theorem}{\bf Theorem}
\begin{document}
\title{Unified Error Analysis for Synchronous and Asynchronous Two-User Random Access}


\author{%
  \IEEEauthorblockN{Nazanin Mirhosseini and Jie Luo}
  \IEEEauthorblockA{ECE Dept., Colorado State Univ., Fort Collins, CO 80523\\
                     Email:\{Nazanin.Mirhosseini, Jie.Luo\}@colostate.edu}
}

\maketitle


\begin{abstract}
We consider a two-user random access system in which each user independently selects a coding scheme from a finite set for every message, without sharing these choices with the other user or with the receiver. The receiver aims to decode only user~1’s message but may also decode user~2’s message when beneficial. In the synchronous setting, the receiver employs two parallel sub-decoders: one dedicated to decoding user~1’s message and another that jointly decodes both users’ messages. Their outputs are synthesized to produce the final decoding or collision decision. For the asynchronous setting, we examine a time interval containing $L$ consecutive codewords from each user. The receiver deploys $2^{2L}$ parallel sub-decoders, each responsible for decoding a subset of the message-code index pairs. In both synchronous and asynchronous cases, every sub-decoder partitions the coding space into three disjoint regions: operation, margin, and collision, and outputs either decoded messages or a collision report according to the region in which the estimated code index vector lies. Error events are defined for each sub-decoder and for the overall receiver whenever the expected output is not produced. We derive achievable upper bounds on the generalized error performance, defined as a weighted sum of incorrect-decoding, collision, and miss-detection probabilities, for both synchronous and asynchronous scenarios\footnote{This work was supported by the National Science Foundation under Grant ECCS-2128569. Any opinions, findings, and conclusions or recommendations expressed in this paper are those of the authors and do not necessarily reflect the views of the National Science Foundation.}.
\end{abstract}

\section{Introduction}
Random access communication manages uncoordinated and bursty transmissions by relying on pre-defined protocols rather than real-time \emph{online} coordination \cite{Bertsekas1992}, \cite{Ephremides1998}. A central challenge in such systems is that a receiver must be prepared to decode a message from any transmitter within a large population of legitimate users. If each user were assigned a unique codebook, identifying an active user would require prohibitively high computational complexity \cite{Gallager1985}, \cite{Polyanskiy2017}, \cite{Tang2018}. Practical systems therefore employ a limited set of pre-designed signaling schemes, enabling the receiver to decode a message without knowing the sender’s identity in advance, an idea formalized by the \emph{unsourced random access} (URA) model \cite{Polyanskiy2017}.

Although recent URA work has been motivated by supporting multi-packet reception (MPR), the standard URA model contains assumptions that are incompatible with this goal. First, URA assumes that user identification (ID) is unnecessary. However, under MPR, when multiple users transmit packet streams concurrently, the receiver cannot properly reassemble messages unless each packet embeds its user ID; this in turn forces distinct effective codebooks and contradicts the single-codebook premise of URA \cite{Ghez1988}, \cite{Ghez1989}. Second, URA relies on a single, fixed codebook for all users, which inherently limits the number of decodable parallel messages. A codebook optimized for a large number of users must operate at a low rate, leading to poor throughput when only a few users are active \cite{Luo2012}, \cite{Luo2015}. Consequently, any fixed coding scheme, if pre-specified in the physical-layer protocol, is poorly suited for the time-varying traffic conditions of random access.

To address these limitations, this paper investigates a \emph{semi-unsourced random access model}. In this framework, each active user independently selects a coding option from a finite, pre-shared ensemble for every message, without coordinating its choice with the receiver or other users. Code ensembles are unique across active users, enabling the receiver to distinguish active transmitters while allowing each user to dynamically adapt parameters such as rate or power based on sensed channel contention. The receiver, knowing the code ensembles but not the selected indices, must jointly identify the coding options and decode the messages of interest. This is achieved using a bank of parallel sub-decoders. In the synchronous scenario, these sub-decoders implement different decoding strategies for recovering the message of a target user (denoted as user~1), such as treating interference as noise or jointly decoding multiple users’ messages. The final decision combines their outputs to obtain user~1’s decoded message or to declare a collision.

We further extend the analysis to an asynchronous scenario to incorporate key practical factors in random access. Protocols such as Carrier Sense Multiple Access with Collision Avoidance (CSMA/CA) require users to sense the channel and postpone transmissions when it is busy \cite{Bianchi2000}, \cite{Luo2024}, leading to deliberate, staggered packet arrivals. Supporting such protocols requires a communication model that accommodates frame-asynchronous codeword transmissions. We therefore analyze a setting in which user transmissions are not aligned in time, ensuring that the proposed framework remains applicable to realistic system designs.

\section{Frame Synchronous Transmission} \label{Sec1}
Consider symbol and frame synchronous random multiple access communication over a memoryless channel with two users and one receiver. The channel is modeled by the conditional distribution $P_{Y|X_1, X_2}$, where $Y\in {\cal Y}$ is the channel output symbol, $X_1\in {\cal X}_1$ and $X_2\in {\cal X}_2$ are respectively the channel input symbols from the two users, with ${\cal Y}$, ${\cal X}_1$, and ${\cal X}_2$ being the corresponding alphabets. For notational convenience, define the input vector $\mbox{\boldmath $X$}=[X_1, X_2]^{\top}$ and denote the channel by $P_{Y|\mbox{\scriptsize \boldmath $X$}}$. Each user is equipped with an ensemble of random blockcodes. A code is denoted by $g=[P_g(X), r_g]$, where $P_g(X)$ specifies the input distribution and $r_g$ is the rate parameter. While codebooks are generated randomly, we assume that their pseudo-random generation algorithms are shared offline with the receiver. If the random seeds (e.g., time and temporary user ID) are available online, the receiver can reproduce the ensemble of codebooks used by the users. Time is slotted, with each slot spanning $N$ channel symbol durations. At the beginning of each time slot, each user selects a code to encode its message and transmits the corresponding codeword. The receiver, unaware of the users’ code selections, searches over possible code combinations and attempts to decode. If the resulting error probability is within a prescribed threshold, the receiver outputs the decoded messages; otherwise, it declares a \textit{collision}.

To be technically specific, consider communication in one time slot. Let $(w_1, g_1)$ and $(w_2, g_2)$ denote the message-code pairs of the two users, where $g_1\in {\cal G}_1$ and $g_2\in {\cal G}_2$. The corresponding codewords are $X^{(N)}_1(w_1, g_1)$ and $X^{(N)}_2(w_2, g_2)$. The receiver knows all randomly generated codebooks in ${\cal G}_1$ and ${\cal G}_2$, but is unaware of the transmitted messages or code selections. We assume that the receiver is only interested in decoding the message of user 1. Given the received channel output $Y^{(N)}$, the receiver either outputs an estimated message-code pair $(\hat{w}_1, \hat{g}_1)$ or declares a collision for user 1.

For convenience, define $\mbox{\boldmath $w$}=[w_1, w_2]^{\top}$, $\hat{\mbox{\boldmath $w$}}=[\hat{w}_1, \hat{w}_2]^{\top}$ and $\mbox{\boldmath $g$}=[g_1, g_2]^{\top}$, $\hat{\mbox{\boldmath $g$}}=[\hat{g}_1, \hat{g}_2]^{\top}$. The receiver partitions the space of coding vectors $\mbox{\boldmath $g$}$ into three disjoint regions: an \textit{operation region} ${\cal R}$, an \textit{operation margin} $\widehat{{\cal R}}$, and a \textit{collision region} ${\cal R}_c$. While the receiver does not know $\mbox{\boldmath $g$}$, if $\mbox{\boldmath $g$}\in {\cal R}$, the receiver attempts decoding; if $\mbox{\boldmath $g$}\in {\cal R}_c$, it attempts a collision report; and if $\mbox{\boldmath $g$}\in\widehat{{\cal R}}$, both correct decoding and collision report are acceptable. We define the conditional error probability of the system as
\begin{equation}
	P_e(\mbox{\boldmath $g$})=\left\{\begin{array}{l} \mathbb{P}\left[(\hat{w}_1, \hat{g}_1)\ne (w_1, g_1)| (\mbox{\boldmath $w$}, \mbox{\boldmath $g$})\right],\   \forall \mbox{\boldmath $g$} \in {\cal R} \\
	1- \mathbb{P}\left[\left.\begin{array}{l} \mbox{``collision'' or } \\ (\hat{w}_1, \hat{g}_1)= (w_1, g_1)\end{array}\right|(\mbox{\boldmath $w$}, \mbox{\boldmath $g$})\right], \\ \qquad \qquad \qquad \qquad \qquad \qquad \ \ \ \ \ \forall \mbox{\boldmath $g$} \in \widehat{{\cal R}} \\
	1- \mathbb{P}\left[\mbox{``collision''}|(\mbox{\boldmath $w$}, \mbox{\boldmath $g$})\right], \qquad \ \forall \mbox{\boldmath $g$} \in {\cal R}_c \end{array}. \right.
\end{equation}
Let $\{P(\mbox{\boldmath $g$})\}$ be a set of pre-determined weight parameters (interpreted as assumed prior probabilities) satisfying $P(\mbox{\boldmath $g$})\ge 0$ for all $\mbox{\boldmath $g$}$ and $\sum_{\mbox{\scriptsize \boldmath $g$}} P(\mbox{\boldmath $g$})=1$. We define the \emph{generalized error performance (GEP)} of the system as
\begin{equation}
 	\mathrm{GEP} =\sum_{\mbox{\scriptsize \boldmath $g$}}  P_e(\mbox{\boldmath $g$})P(\mbox{\boldmath $g$})= \mathbb{E}_{\mbox{\scriptsize \boldmath $g$}}[P_e(\mbox{\boldmath $g$})].
\label{ExtendedGEP}
\end{equation}

Although the receiver is concerned only with user 1's message, it may choose to jointly decode user 2's message if this aids decoding. Particularly, the receiver employs two sub-decoders in parallel:
\begin{itemize}
   \item $D_{12}$: jointly decodes the messages of both users~1 and~2,
   \item $D_{1}$: decodes only the message of user~1 without attempting to decode user~2.
\end{itemize}
The final decoding or collision decision is obtained by combining the outcomes of these two sub-decoders, if they are consistent.

\subsection{Sub-Decoder $D_{12}$}

Sub-decoder $D_{12}$ should set its operation region, operation margin and collision region, respectively denoted by $\mathcal{R}_{12}$, $\widehat{\mathcal{R}}_{12}$, and $\mathcal{R}_{12c}$ as follows
\begin{eqnarray}
	&& \mathcal{R}_{12}  \subseteq \mathcal{R}, \quad \widehat{\mathcal{R}}_{12} =  \widehat{\mathcal{R}} \cup (\mathcal{R} \backslash \mathcal{R}_{12}), \quad \mathcal{R}_{12c} = \mathcal{R}_{c}.
\end{eqnarray}
Given the actual coding vector $\mbox{\boldmath $g$}$, define the conditional error probability for sub-decoder $D_{12}$ as
\begin{equation}
	P_{e12}(\mbox{\boldmath $g$}):=\left\{\begin{array}{l} \mathbb{P}\left[\left. \begin{array}{l} (\hat{w}_1, \hat{g}_1)\ne (w_1, g_1) \mbox{ or} \\ (\hat{w}_2, \hat{g}_2)\ne (w_2, g_2) \end{array} \right| (\mbox{\boldmath $w$}, \mbox{\boldmath $g$})\right], \\  \qquad \qquad \qquad \qquad \qquad \qquad \qquad 	\forall \mbox{\boldmath $g$} \in {\cal R}_{12} \\
	1- \mathbb{P}\left[\left.\begin{array}{l} \mbox{``collision'' or } \\ (\hat{\mbox{\boldmath $w$}},\hat{\mbox{\boldmath $g$}})= (\mbox{\boldmath $w$},\mbox{\boldmath $g$})  \end{array}\right|	(\mbox{\boldmath $w$}, \mbox{\boldmath $g$})\right], \\ \qquad \qquad \qquad \qquad \qquad \qquad \qquad \forall \mbox{\boldmath $g$} \in \widehat{{\cal R}}_{12} \\
	1- \mathbb{P}\left[\mbox{``collision''}|(\mbox{\boldmath $w$}, \mbox{\boldmath $g$})\right], \qquad \ \ \ \forall \mbox{\boldmath $g$} \in {\cal R}_{12c} \end{array} \right. \label{DefP_e12}
\end{equation}
Consequently, the \emph{generalized error performance} for $D_{12}$ is defined as
\begin{eqnarray}
	\mathrm{GEP}_{12} := \sum_{\mbox{\scriptsize \boldmath $g$}\in  \mathcal{G}_{1}\times \mathcal{G}_{2}}P_{e12}(\mbox{\boldmath $g$})P(\mbox{\boldmath $g$}),\label{GEP12Def}
\end{eqnarray}

The following theorem establishes an achievable bound on $\mbox{GEP}_{12}$.
\begin{theorem}\label{GEP12}
Consider the sub-decoder $D_{12}$, which jointly decodes both users~$1$ and $2$ messages. The generalized error performance $\mbox{GEP}_{12}$ for this sub-decoder is upper bounded by
\begin{eqnarray}
	&& \mbox{GEP}_{12} \le \sum_{\mbox{\scriptsize \boldmath $g$} \in {\cal R}_{12}}\left(\sum_{\tilde{\mbox{\scriptsize \boldmath $g$}} \in {\cal R}_{12}, \tilde{g}_1=g_1} B_{i\{1\}}^{(D_{12})}(\tilde{\mbox{\boldmath $g$}},\mbox{\boldmath $g$}) \right. \nonumber \\
	&& +\sum_{\tilde{\mbox{\scriptsize \boldmath $g$}} \in {\cal R}_{12}, \tilde{g}_2=g_2} B_{i\{2\}}^{(D_{12})}(\tilde{\mbox{\boldmath $g$}},\mbox{\boldmath $g$}) + \sum_{\tilde{\mbox{\scriptsize \boldmath $g$}} \in {\cal R}_{12}} \ \ \ \ B_{i\{\}}^{(D_{12})}(\tilde{\mbox{\boldmath $g$}},\mbox{\boldmath $g$}) \nonumber \\
	&& + 2 \!\!\!\!\!\!\!\! \sum_{\tilde{\mbox{\scriptsize \boldmath $g$}} \in \widehat{{\cal R}}_{12}\cup {\cal R}_{12c}, \tilde{g}_1=g_1} \!\!\!\!\!\!\!\!  B_{mc\{1\}}^{(D_{12})}(\tilde{\mbox{\boldmath $g$}},\mbox{\boldmath $g$}) + 2 \!\!\!\!\!\!\!\!  \sum_{\tilde{\mbox{\scriptsize \boldmath $g$}} \in \widehat{{\cal R}}_{12}\cup {\cal R}_{12c}, \tilde{g}_2=g_2} \!\!\!\!\!\!\!\!  B_{mc\{2\}}^{(D_{12})}(\tilde{\mbox{\boldmath $g$}},\mbox{\boldmath $g$}) \nonumber \\
	&& \left. + 2 \!\!\!\!\!  \sum_{\tilde{\mbox{\scriptsize \boldmath $g$}} \in \widehat{{\cal R}}_{12}\cup {\cal R}_{12c}}  \!\!\!\!\!\!\!\! B_{mc\{\}}^{(D_{12})}(\tilde{\mbox{\boldmath $g$}},\mbox{\boldmath $g$})    \right) \nonumber \\
\label{2UserJointGEPBound}
\end{eqnarray}
Here, the bounds on incorrect decoding probabilities are
\begin{eqnarray}
	&& \!\!\!\!\!\!\!\!\!\!\!\!  B_{i\{1\}}^{(D_{12})}(\tilde{\mbox{\boldmath $g$}}, \mbox{\boldmath $g$})\! = \! P(\mbox{\boldmath $g$})\mathbb{E}\Bigg[\min \Bigg\{\frac{1}{|\mathcal{R}_{12g_1}|}, 2^{Nr_{\tilde{g}_{2}}}\times  \\
	&& \!\!\!\!\!\!\!\!\!\!\!\!\!  \mathbb{P}\! \left[\ell_{\tilde{\mbox{\scriptsize \boldmath $g$}}}\!\left(X_{1}^{(N)}\! ,\tilde{X}_{2}^{(N)}\! ,Y^{(N)}\! \right)\! >\!  \ell_{\mbox{\scriptsize \boldmath $g$}}\!\! \left(\mbox{\boldmath $X$}^{(N)}\! ,Y^{(N)}\! \right) \! \middle \vert  \mbox{\boldmath  $X$}^{(N)}\! , Y^{(N)}\! \right]\!\! \Bigg\}\! \Bigg],\nonumber \\
	&& \!\!\!\!\!\!\!\!\!\!\!\! B_{i\{2\}}^{(D_{12})}(\tilde{\mbox{\boldmath $g$}}, \mbox{\boldmath $g$})= P(\mbox{\boldmath $g$})\mathbb{E}\Bigg[\min \Bigg\{\frac{1}{|\mathcal{R}_{12g_2}|}, 2^{Nr_{\tilde{g}_{1}}}\times \\
	&& \!\!\!\!\!\!\!\!\!\!\!\!\!  \mathbb{P}\!\left[\ell_{\tilde{\mbox{\scriptsize \boldmath $g$}}}\!\left(\tilde{X}_{1}^{(N)}\! ,{X}_{2}^{(N)}\! ,Y^{(N)}\! \right)\! >\!  \ell_{\mbox{\scriptsize \boldmath $g$}}\!\! \left(\mbox{\boldmath $X$}^{(N)}\! ,Y^{(N)}\! \right) \! \middle \vert  \mbox{\boldmath  $X$}^{(N)}\! , Y^{(N)}\! \right]\!\! \Bigg\}\! \Bigg],\nonumber \\
	&& \!\!\!\!\!\!\!\!\!\!\!\!  B_{i\{\}}^{(D_{12})}(\tilde{\mbox{\boldmath $g$}}, \mbox{\boldmath $g$}) \! = \! P(\mbox{\boldmath $g$}) \mathbb{E}\Bigg[\min \Bigg\{ \frac{1}{| \mathcal{R}_{12}|}, 2^{N(r_{\tilde{g}_{1}}+r_{\tilde{g}_{2}})}  \times \\
	&& \!\!\!\!\!\!\!\!\!\!\!\!\!  \mathbb{P}\left[\ell_{\tilde{\mbox{\scriptsize \boldmath $g$}}}\! \left(\tilde{\mbox{\boldmath $X$}}^{(N)}\! , Y^{(N)}\! \right) \! >\! \ell_{\mbox{\scriptsize \boldmath $g$}}\! \left(\mbox{\boldmath $X$}^{(N)},Y^{(N)}\! \right) \! \middle  \vert \mbox{\boldmath $X$}^{(N)},Y^{(N)}\right]\!\! \Bigg\}\! \Bigg], \nonumber
\end{eqnarray}
and the bounds on miss-detection and collision probabilities are
\begin{eqnarray}
	&&  \!\!\!\!\!\!\!\!\!\!\!\!\!\!\!\! B_{mc\{1\}}^{(D_{12})}(\tilde{\mbox{\boldmath $g$}}, \mbox{\boldmath $g$})=  \\
	&& \!\!\!\!\!\!\!\!\!\!\!\!\! \min\Bigg\{ P(\mbox{\boldmath $g$}), P(\tilde{\mbox{\boldmath $g$}})2^{Nr_{g_2}} \mathbb{E}\Bigg[\frac{P_{\tilde{g}_{2}}\left(Y^{(N)}\middle \vert X_{1}^{(N)}\right)}{P\left(Y^{(N)}\middle \vert \mbox{\boldmath  $X$}^{(N)}\right)}\Bigg] \Bigg\},\nonumber \\
	&& \!\!\!\!\!\!\!\!\!\!\!\!\!\!\!\! B_{mc\{2\}}^{(D_{12})}(\tilde{\mbox{\boldmath $g$}}, \mbox{\boldmath $g$})=  \\
	&& \!\!\!\!\!\!\!\!\!\!\!\!\! \min\Bigg\{ P(\mbox{\boldmath $g$}), P(\tilde{\mbox{\boldmath $g$}})2^{Nr_{g_1}}\mathbb{E}\Bigg[\frac{P_{\tilde{g}_{1}}\left(Y^{(N)}\middle \vert X_{2}^{(N)}\right)}{P\left(Y^{(N)}\middle \vert \mbox{\boldmath  $X$}^{(N)}\right)}\Bigg] \Bigg\},\nonumber \\
	&& \!\!\!\!\!\!\!\!\!\!\!\!\!\!\!\! B_{mc\{\}}^{(D_{12})}(\tilde{\mbox{\boldmath $g$}}, \mbox{\boldmath $g$})=  \\
	&& \!\!\!\!\!\!\!\!\!\!\!\!\! \min \Bigg\{ P(\mbox{\boldmath $g$}), P(\tilde{\mbox{\boldmath $g$}}) 2^{N(r_{g_1}+r_{g_2})} \mathbb{E}\Bigg[\frac{P_{\tilde{\mbox{\scriptsize \boldmath $g$}}}\left(Y^{(N)}\right)}{P\left(Y^{(N)}\middle \vert \mbox{\boldmath $X$}^{(N)}\right)}\Bigg]\! \Bigg\},\nonumber
\end{eqnarray}
where $|.|$ is the cardinality of a set, the likelihood is denoted by
\begin{eqnarray}
	\ell_{\mbox{\scriptsize \boldmath $g$}}\! \left(\mbox{\boldmath  $X$}^{(N)},Y^{(N)}\right)\! =  \! P\! \left(Y^{(N)}\middle \vert \mbox{\boldmath $X$}^{(N)}\! \right)\! P(\mbox{\boldmath $g$}) 2^{-N(r_{g_1}+r_{g_2})}, \nonumber
\end{eqnarray}
and $\mathcal{R}_{12g_1} = \{ \tilde{\mbox{\boldmath $g$}} \in \mathcal{R}_{12} : \tilde{g}_{1} = g_1\}$, $\mathcal{R}_{12g_2} = \{ \tilde{\mbox{\boldmath $g$}} \in \mathcal{R}_{12} : \tilde{g}_{2} = g_2\}$.
\end{theorem}
\begin{proof}
	The proof of Theorem \ref{GEP12} is provided in Appendix \ref{AppGEP12}.
\end{proof}

\subsection{Sub-Decoder $D_1$}

We now apply the same argument used for sub-decoder $D_{12}$ to sub-decoder $D_1$, considering its decoding region to be defined as follows
\begin{eqnarray}
	&&\mathcal{R}_{1} \subseteq \mathcal{R}  \backslash \mathcal{R}_{12}, \quad \widehat{\mathcal{R}}_{1} = \widehat{\mathcal{R}}\cup \mathcal{R}_{12}, \quad \mathcal{R}_{1c}  = \mathcal{R}_{c}.
\end{eqnarray}
Note that operations regions $\mathcal{R}_{12}$ and $\mathcal{R}_{1}$ form a partition of $\mathcal{R}$. Given the actual coding vector $\mbox{\boldmath $g$}$, the conditional error probability for $D_1$ is defined as
\begin{equation}
	P_{e1}(\mbox{\boldmath $g$})=\left\{\begin{array}{l} \mathbb{P}\left[\left. (\hat{w}_1, \hat{g}_1)\ne (w_1, g_1) \right| (\mbox{\boldmath $w$}, \mbox{\boldmath $g$})\right],  \ \forall \mbox{\boldmath $g$}\in {\cal R}_{1} \\
	1- \mathbb{P}\left[\left.\begin{array}{l} \mbox{``collision'' or } \\ (\hat{w}_1, \hat{g}_1)= (w_1, g_1) \end{array}\right|(\mbox{\boldmath $w$}, \mbox{\boldmath $g$})\right], \\ \qquad \qquad \qquad \qquad \qquad \qquad \qquad \forall \mbox{\boldmath $g$} \in \widehat{{\cal R}}_{1} \\
	1- \mathbb{P}[\mbox{``collision''}|(\mbox{\boldmath $w$}, \mbox{\boldmath $g$})], \ \ \ \ \ \  \ \ \ \ \forall \mbox{\boldmath $g$} \in {\cal R}_{1c} \end{array} \right.
\end{equation}
The \emph{generalized error performance} for sub-decoder $D_1$ is defined analogous to (\ref{GEP12Def}) as
\begin{eqnarray}
	\mathrm{GEP}_{1} := \sum_{\mbox{\scriptsize \boldmath $g$}\in \mathcal{G}_{1} \times \mathcal{G}_{2}} P_{e1}(\mbox{\boldmath $g$})P(\mbox{\boldmath $g$}).
\end{eqnarray}

The following theorem provides an achievable bound on $\mbox{GEP}_{1}$.
\begin{theorem}\label{GEP1}
Consider sub-decoder $D_1$, which only decodes the message of user~$1$ without attempting to decode user~$2$'s message. The generalized error performance $\mbox{GEP}_{1}$ of this sub-decoder is upper bounded by
\begin{eqnarray}
	&& \mbox{GEP}_{1} \le \sum_{\mbox{\scriptsize \boldmath $g$} \in {\cal R}_{1}}\left(\sum_{\tilde{\mbox{\scriptsize \boldmath $g$}} \in {\cal R}_{1}} B_{i\{\}}^{(D_1)}(\tilde{\mbox{\boldmath $g$}}, \mbox{\boldmath $g$}) \right. \nonumber \\
	&& + 2  \!\!\!\!\!  \sum_{\tilde{\mbox{\scriptsize \boldmath $g$}} \in \widehat{{\cal R}}_{1}\cup {\cal R}_{1c}}  \!\!\!\!\!\!\!\!  B_{mc\{\}}^{(D_1)}(\tilde{\mbox{\boldmath $g$}},\mbox{\boldmath $g$})  \left. +2   \!\!\!\!\! \sum_{\tilde{\mbox{\scriptsize \boldmath $g$}} \in {\cal R}_{1c}, \tilde{g}_1=g_1}  \!\!\!\!\!\!\!\! B_{mc\{1\}}^{(D_1)}(\tilde{\mbox{\boldmath $g$}},\mbox{\boldmath $g$}) \right).
\label{2UserJointGEPBound}
\end{eqnarray}
Here, the bound on the probability of incorrect decoding is
\begin{eqnarray}
	&& \!\!\!\!\!\!\! B_{i\{\}}^{(D_1)}(\tilde{\mbox{\boldmath $g$}}, \mbox{\boldmath $g$})= P(\mbox{\boldmath $g$}) \mathbb{E}\Bigg[\min \Bigg\{ \frac{1}{|\mathcal{R}_{1}|}, 2^{Nr_{\tilde{g}}}\times \\
	&& \!\!\!\!\!\!\!\! \mathbb{P}\left[ \ell_{\tilde{\mbox{\scriptsize \boldmath $g$}}}\left(\tilde{X}_{1}^{(N)},Y^{(N)}\right)\! > \! \ell_{\mbox{\scriptsize \boldmath $g$}}\left({X}_{1}^{(N)},Y^{(N)}\right)  \middle \vert X_{1}^{(N)},  Y^{(N)}\right]\bigg\}\Bigg],\nonumber
\end{eqnarray}
and the bounds on miss-detection and collision probabilities are
\begin{eqnarray}
	&& \!\!\!\!\!\!\!\!\!\!\! B_{mc\{\}}^{(D_1)}(\tilde{\mbox{\boldmath $g$}}, \mbox{\boldmath $g$})= \\
	&& \!\!\!\!\!\!\!\!\!\!\!  \min \Bigg\{ P(\mbox{\boldmath $g$}), P(\tilde{\mbox{\boldmath $g$}}) 2^{Nr_{g_1}}\mathbb{E}\Bigg[\frac{P_{\tilde{\mbox{\scriptsize \boldmath $g$}}}\left(Y^{(N)}\right)}{P\left(Y^{(N)}\middle \vert X_{1}^{(N)},g_2\right)}\Bigg]\Bigg\},\nonumber \\
	&& \!\!\!\!\!\!\!\!\!\!\!\! B_{mc\{1\}}^{(D_1)}(\tilde{\mbox{\boldmath $g$}}, \mbox{\boldmath $g$}) = \\
    && \!\!\!\!\!\!\!\!\!\!\! \min \Bigg\{ P(\mbox{\boldmath $g$}), P(\tilde{\mbox{\boldmath $g$}}) \mathbb{E}\Bigg[\frac{P_{\tilde{\mbox{\scriptsize \boldmath $g$}}}\left(Y^{(N)}\middle \vert X_{1}^{(N)}\right)}{P\left(Y^{(N)}\middle \vert X_{1}^{(N)},g_2\right)}\Bigg]\Bigg\}, \nonumber
\end{eqnarray}
where $|.|$ denotes the cardinality of a set, and
\begin{eqnarray}
	\ell_{\mbox{\scriptsize \boldmath $g$}}\left(X_1^{(N)},Y^{(N)}\right) = P\left(Y^{(N)} \middle  \vert X_1^{(N)},g_2\right)P(\mbox{\boldmath $g$})2^{-Nr_{g_1}}. \nonumber
\end{eqnarray}
\begin{proof}
	The proof of Theorem \ref{GEP1} is provided in Appendix \ref{AppGEP1}.
\end{proof}
\end{theorem}

\subsection{Ultimate Output of Receiver}
\label{UltimateOutput}
The receiver's final decision is determined by combining the outputs from sub-decoders $D_1$ and $D_{12}$ based on four possible scenarios. If both sub-decoders agree on the same message-code index pair for user~$1$, the receiver outputs this pair, $(\hat{w}_{1},\hat{g}_{1})$. Conversely, if both output a message-code index pair but they are different, the receiver reports a collision. In the event that only one sub-decoder outputs a message-code index pair while the other reports a collision, the receiver adopts the non-collision output. Finally, if both sub-decoders report a collision, the receiver's ultimate output is also a collision.

\begin{theorem}\label{THM}
	Consider the described receiver containing sub-decoders $D_{12}$ and $D_1$ with respective generalized error performance $\mathrm{GEP}_{12}$ and $\mathrm{GEP}_{1}$. The system's generalized error performance $\mathrm{GEP}$ is then bounded by
	\begin{equation}
		\mathrm{GEP} \leq \max_{\substack{\mathcal{R}_{1},\mathcal{R}_{12} \\ \mathrm{s.t.} \ \mathcal{R}_{1}\cup \mathcal{R}_{12}=\mathcal{R}, \mathcal{R}_{1}\cap \mathcal{R}_{12}=\emptyset}} \left(\mathrm{GEP}_{1} + \mathrm{GEP}_{12}\right).
	\end{equation}
\end{theorem}

The proof of Theorem \ref{THM} is omitted.

\textit{Remark 1}. Note that the bounds derived in Theorems \ref{GEP12} and \ref{GEP1} can be numerically evaluated using Gaussian and saddle point approximation \cite{SaddlePoint2020}, \cite{SaddlePoint2010}.

\section{Frame Asynchronous Transmission}
Classical random access communication often employs the CSMA/CA protocol to reduce collisions and improve throughput. In CSMA/CA, a user senses channel availability during a short interval, termed the DCF Interframe Space (DIFS)~ \cite{Bianchi2000}. If the channel is sensed idle, the user randomly decides whether to transmit; otherwise, the user remains idle and continues sensing~ \cite{Bianchi2000}. A key requirement for this protocol to be effective is that the DIFS should be much shorter than the packet duration~\cite{Bianchi2000}, \cite{Luo2024}, so that sensing results during one DIFS provide predictive information about near-future channel availability.

When parallel decoding of multiple users is supported, the channel may still be considered available even while a codeword transmission is in progress. Consequently, a user may begin transmitting after sensing channel availability, resulting in frame-asynchronous transmissions relative to ongoing codewords. To properly accommodate CSMA/CA, it is therefore important for random access models and the corresponding channel coding theory to explicitly support frame-asynchronous codeword transmissions. Since such asynchrony is deliberate, it can be partially controlled by system design. For example, each time slot may be partitioned into several equal-sized mini-slots, and users may be required to start transmissions only at mini-slot boundaries. Although transmissions are asynchronous in this case, the receiver retains partial knowledge of possible transmission start times, which aids in detecting new codeword arrivals.

\subsection{Asynchronous System Model and Notations}

To keep notation manageable, we focus on a simplified scenario. Suppose two users transmit in parallel, as illustrated in Fig.~\ref{FigAsynchronous}. The receiver can detect the beginning and end of continuous transmissions. Specifically, continuous transmission starts with user~1 at symbol $0$ and ends with user~2 at symbol $NL+T_2$, where $N$ is the codeword length, $L$ is a positive integer representing the number of consecutive packets in one frame, and $T_2$ is the transmission offset of user~2.
\begin{figure}[!h]
   \begin{center}
       \includegraphics[width=3.5 in]{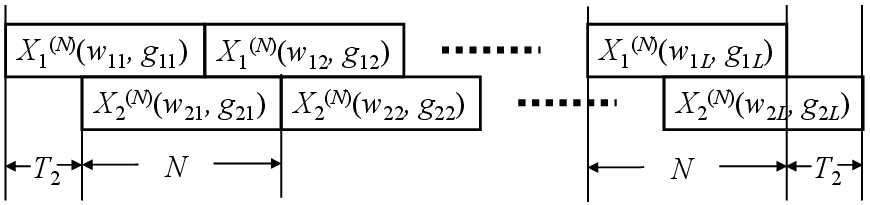}
       \caption{\label{FigAsynchronous} Asynchronous transmission from two users.}
   \end{center}
\end{figure}
We assume a static channel state, with both the channel realization and the value of $T_2$ known to the receiver. If these values are unknown, the framework can still be accommodated using the virtual-user construction introduced in \cite{Luo2015}.

Let $w_{ij}$ denote the $j$-th message of user $i$, where $j=1,\dots,L$ and $i\in\{1,2\}$. The corresponding coding option and codeword are denoted by $g_{ij}$ and $X^{(N)}_i(w_{ij}, g_{ij})$, respectively. Note that some coding options may correspond to the ``idling'' option. The channel input sequences from user~1 and user~2 are
\begin{eqnarray}
&& \mbox{\boldmath $X$}_{1}^{(N)}=[X^{(N)}_1(w_{11}), \dots, X^{(N)}_1(w_{1L}), \mbox{idling}^{(T_2)}], \nonumber \\
&& \mbox{\boldmath $X$}_{2}^{(N)}=[\mbox{idling}^{(T_2)}, X^{(N)}_2(w_{21}), \dots, X^{(N)}_2(w_{2L})]. \nonumber
\end{eqnarray}
The channel output sequence is $Y^{(NL+T_2)}$. These sequences are related by the conditional probability
\begin{equation*}
P\left(Y^{(NL+T_2)}| \mbox{\boldmath $X$}_{1}^{(N)}, \mbox{\boldmath $X$}_{2}^{(N)}\right),
\end{equation*}
derived from the memoryless channel model $P_{Y|X_1,X_2}$.

Define the universal set of message indices as
\begin{equation*}
U=\left \{[i,j]:i\in\{1,2\}, j\in\{ 1, \dots, L\}\right \}.
\end{equation*}
Let $\mbox{\boldmath $g$}_{U}$ be the coding vector containing all $g_{ij}$ for $[i,j]\in U$, and let $\mbox{\boldmath $w$}$ be the vector of all corresponding messages. Following the analysis in Section~\ref{Sec1}, the receiver typically employs multiple sub-decoders operating in parallel, where each sub-decoder jointly decodes a subset of messages without decoding the others.

\subsection{$\mathrm{GEP}$ for Asynchronous scenario}

Consider one such sub-decoder, denoted $D_D$, where $D \subseteq U$ is the subset of indices targeted for decoding. Denote by $\mbox{\boldmath $g$}_D$ and $\mbox{\boldmath $w$}_D$ the coding and message vectors restricted to indices in $D$. Sub-decoder $D_D$ either outputs estimates $(\hat{\mbox{\boldmath $w$}}_D, \hat{\mbox{\boldmath $g$}}_D)$ or reports ``collision.'' Let ${\cal R}_D$, $\widehat{{\cal R}}_D$, and ${\cal R}_{Dc}$ denote its operation region, operation margin, and collision region, respectively. Given the actual coding vector $\mbox{\boldmath $g$}_{U}$, the conditional error probability of $D_D$ is defined as
\begin{equation}
	P_{eD}(\mbox{\boldmath $g$}_{U})=\left\{\begin{array}{l} \mathbb{P}\left[\left. \begin{array}{l} \exists [i,j]\in D \mbox{ such that} \\ (\hat{w}_{ij}, \hat{g}_{ij})	\ne (w_{ij}, g_{ij}) \end{array} \right| (\mbox{\boldmath $w$}, \mbox{\boldmath $g$}_{U})\right], \\  \qquad \qquad \qquad \qquad \qquad \qquad \ \ \ \ \forall \mbox{\boldmath $g$}_{U} \in {\cal R}_D \\
	1- \mathbb{P}\left[\left.\begin{array}{l} \mbox{``collision'' or } \\ (\hat{\mbox{\boldmath $w$}}_D, \hat{\mbox{\boldmath $g$}}_D)= (\mbox{\boldmath $w$}_D, 	\mbox{\boldmath $g$}_D) \end{array}\right|(\mbox{\boldmath $w$}, \mbox{\boldmath $g$}_{U})\right], \\ \qquad \qquad \qquad \qquad \qquad \qquad \ \ \ \  \forall 	\mbox{\boldmath $g$}_{U} \in \widehat{{\cal R}}_D \\
	1- \mathbb{P}[\mbox{``collision''}|(\mbox{\boldmath $w$}, \mbox{\boldmath $g$}_{U})], \ \ \ \ \ \  \forall \mbox{\boldmath $g$}_{U} \in {\cal R}_{Dc} \end{array} \right.
\end{equation}
With a set of predetermined weight parameters $\{P(\mbox{\boldmath $g$}_{U})\}$ satisfying $P(\mbox{\boldmath $g$}_{U}) \geq 0$ and $\sum_{\mbox{\scriptsize \boldmath $g$}_{U}} P(\mbox{\boldmath $g$}_{U}) = 1$, we define the ``generalized error performance'' of $D_D$ as
\begin{equation}
	\mbox{GEP}_D =\sum_{\mbox{\scriptsize \boldmath $g$}_{U}}  P_{eD}(\mbox{\boldmath $g$}_{U})P(\mbox{\boldmath $g$}_{U}).
\label{ExtendedGEPM}
\end{equation}

The following theorem provides an achievable bound on $\mbox{GEP}_D$.

\begin{theorem}\label{Theorem1UserDecodingD}
Consider the random multiple access system described in this section, with pre-registered users. The $\mbox{GEP}_D$ of sub-decoder $D_D$ is upper bounded by
\begin{eqnarray}
	&& \!\!\!\!\!\!\!\!\!\!\! \mbox{GEP}_D \le \sum_{\mbox{\scriptsize \boldmath $g$}_{U} \in {\cal R}_D}\left(\sum_{S \subset D } \Biggl( \right. \sum_{\tilde{\mbox{\scriptsize \boldmath $g$}}_{U} \in {\cal R}_D, \tilde{\mbox{\scriptsize \boldmath $g$}}_S=\mbox{\scriptsize \boldmath $g$}_S} B_{iS}(\tilde{\mbox{\boldmath  $g$}}_{U},\mbox{\boldmath $g$}_{U}) \nonumber\\
    && + 2 \sum_{\tilde{\mbox{\scriptsize \boldmath $g$}}_{U} \in \widehat{{\cal R}}_D\cup {\cal R}_{Dc}, \tilde{\mbox{\scriptsize \boldmath $g$}}_S=\mbox{\scriptsize \boldmath $g$}_S} B_{mcS}(\tilde{\mbox{\boldmath $g$}}_{U}, \mbox{\boldmath $g$}_{U}) \Biggr) \nonumber \\
	&& \left. +2 \sum_{\tilde{\mbox{\scriptsize \boldmath $g$}}_{U} \in {\cal R}_{Dc}, \tilde{\mbox{\scriptsize \boldmath $g$}}_D=\mbox{\scriptsize \boldmath $g$}_D} B_{mcD}(\tilde{\mbox{\boldmath $g$}}_{U}, \mbox{\boldmath $g$}_{U}))  \right) ,
\label{2UserJointGEPBound}
\end{eqnarray}
where the bounds for $S \subset D $ are given by
\begin{eqnarray}
	&& \!\!\!\!\!\!\!\!\!\!\! B_{iS}(\tilde{\mbox{\boldmath $g$}}_{U}, \mbox{\boldmath $g$}_{U})=P(\mbox{\boldmath $g$}_{U}) \mathbb{E}\Bigg[\min \Bigg\{\frac{1}{|\mathcal{R}_{D_{\mbox{\scriptsize \boldmath $g$}_{S}}}|}, 2^{N\sum_{[i,j]\in D \setminus S}r_{\tilde{g}_{ij}}} \nonumber \\
    && \times \mathbb{P}\Big[\ell_{\tilde{\mbox{\scriptsize \boldmath $g$}}_{U}}\left(\mbox{\boldmath $X$}_{S}^{(N)}, \tilde{\mbox{\boldmath $X$}}_{D \setminus S}^{(N)}, Y^{(NL+T_2)}\right) > \nonumber \\
	&&   \ell_{\mbox{\scriptsize \boldmath $g$}_{U}}\left(\mbox{\boldmath $X$}_{D}^{(N)}, Y^{(NL+T_2)}\right) \Big \vert \mbox{\boldmath $X$}_{D}, Y^{(NL+T_2)}\Big]\Bigg\}\Bigg], \\
	&& \!\!\!\!\!\!\!\!\!\!\!\! B_{mcS}(\tilde{\mbox{\boldmath $g$}}_{U}, \mbox{\boldmath $g$}_{U})=\min \Bigg\{P(\mbox{\boldmath $g$}_{U}),P(\tilde{\mbox{\boldmath $g$}}_{U}) 2^{N\sum_{[i,j]\in D \setminus S}r_{g_{ij}}} \nonumber \\
    && \times \mathbb{E}\Bigg[\frac{P_{\tilde{\mbox{\scriptsize \boldmath $g$}}_{U}}\left(Y^{(NL+T_2)}\middle \vert \mbox{\boldmath $X$}_{S}^{(N)}\right)}{P\left(Y^{(NL+T_2)}\middle \vert \mbox{\boldmath $X$}_{D}^{(N)},\mbox{\boldmath $g$}_{U \backslash D}\right)}\Bigg]\Bigg\},
\end{eqnarray}
and
\begin{eqnarray}
	&& \!\!\!\!\!\!\!\!\!\!\!\! B_{mcD}(\tilde{\mbox{\boldmath $g$}}_{U}, \mbox{\boldmath $g$}_{U})= \min \Bigg\{P(\mbox{\boldmath $g$}_{U}), \nonumber \\
    && P(\tilde{\mbox{\boldmath $g$}}_{U})\mathbb{E}\Bigg[\frac{P_{\tilde{\mbox{\scriptsize \boldmath $g$}}_{U}}\left(Y^{(NL+T_2)}\middle \vert \mbox{\boldmath $X$}_{D}^{(N)}\right)}{P\left(Y^{(NL+T_2)}\middle \vert \mbox{\boldmath $X$}_{D}^{(N)},\mbox{\boldmath $g$}_{U \backslash D}\right)}\Bigg]\Bigg\}.
\end{eqnarray}
Here, the weighted likelihood is
\begin{eqnarray}
	&& \!\!\!\!\!\! \ell_{\mbox{\scriptsize \boldmath $g$}_{U}}\left( \mbox{\boldmath $X$}_{D}^{(N)}, Y^{(NL+T_2)}\right) = \nonumber \\
    && \!\!\!\!\!\! P\left(Y^{(NL+T_2)} \middle \vert \mbox{\boldmath  $X$}_{D}^{(N)}, \mbox{\boldmath $g$}_{U \backslash D}\right)P(\mbox{\boldmath $g$}_{U})  2^{-N\sum_{[i,j]\in D}r_{g_{ij}}}, \nonumber \\
\end{eqnarray}
and
\begin{eqnarray}
	\mathcal{R}_{D_{\mbox{\scriptsize \boldmath $g$}_{S}}} = \left\{ \tilde{\mbox{\boldmath $g$}}_{S} \in \mathcal{R}_{D} : \tilde{\mbox{\boldmath $g$}}_{S}=\mbox{\boldmath $g$}_{S}\right\}.
\end{eqnarray}
\end{theorem}

\begin{proof}
	The proof of Theorem \ref{Theorem1UserDecodingD} is provided in Appendix \ref{AppGEPD}.
\end{proof}

\textit{Remark 2}. With frame-asynchronous transmissions, a codeword from one user may partially overlap with several codewords from the other user due to transmission offsets. To numerically evaluate the bounds in Theorem~\ref{Theorem1UserDecodingD}, one can apply approximation methods such as Gaussian or saddlepoint approximations \cite{SaddlePoint2010}, \cite{SaddlePoint2020} to each of the overlapping portions. Each bound can also be further tightened using Gallager’s approach \cite{Gallager1965}, \cite{Tang2018}, applied separately to each overlapping portion of a codeword.

Finally, the receiver obtains the desired message–code index pairs by synthesizing the outputs of the sub-decoders, following the approach introduced in Section~\ref{UltimateOutput}. A corresponding upper bound on the final GEP can then be derived.

\section{Discussion on Temporary User Identification}

Random access systems involve uncoordinated users transmitting short and bursty messages, so the receiver does not know in advance which users are active in a given time slot. Because the population of legitimate users is typically massive, assigning each user a distinct codebook would make active-user identification via codebook search computationally prohibitive. At the same time, distinguishing among active transmitters is essential, particularly when multi-packet reception is supported. This can be accomplished in two ways.

The first approach is to maintain a pool of temporary user IDs, denoted by $\mathcal{ID}$, with cardinality $|\mathcal{ID}|$. Each active user randomly selects a temporary ID to distinguish itself from other users in the area. An ID collision occurs when multiple active users choose the same temporary ID. If the number of active users is $K$, the ID collision probability $P_{\mbox{\scriptsize IDC}}$ satisfies \cite{Polyanskiy2017}
\begin{equation}
	P_{\mbox{\scriptsize IDC}} \le \binom{K}{2}\!/|\mathcal{ID}|.
\end{equation}
This probability should be added to the ultimate GEP bounds derived above.

Alternatively, the system may require users to register their selected temporary IDs with a server prior to transmission, resolving any ID conflicts during registration. Although this mechanism deviates slightly from a pure random-access model, the additional overhead is minimal. Moreover, user registration can greatly reduce the codebook search complexity at the receiver, since the receiver can focus only on the registered set of potential active users.

\section{Conclusion}
In this paper, we developed a decoding procedure for a two-user random access channel, addressing both synchronous and asynchronous transmission scenarios. In the synchronous case, we employed two parallel sub-decoders, $D_1$ and $D_{12}$, each partitioning the coding space $\mathcal{G}_{1}\times \mathcal{G}_{2}$ into three disjoint regions: operation, margin, and collision, such that $\mathcal{R}_{12} \cap \mathcal{R}_{1} = \emptyset$. Achievable upper bounds on the corresponding performance measures, $\mathrm{GEP}_{12}$ and $\mathrm{GEP}_{1}$, were derived. For the asynchronous scenario, we extended the analysis to incorporate sequential codeword transmissions, assuming that only user~2 experiences a delay $T_2$. Because the analysis for all $D \subseteq U$ follows the same structure, we focused on deriving achievable upper bounds for a representative sub-decoder $D_D$.

\appendices
\section{Proof of Theorem \ref{GEP12}} \label{AppGEP12}
\begin{proof}
To simplify notation, we abbreviate $X_{i}^{(N)}(w_i,g_i)$ and $X_{i}^{(N)}(\tilde{w}_i,\tilde{g}_i)$ as $X_{i}^{(N)}$ and $\tilde{X}_{i}^{(N)}$, respectively, for $i=1,2$. The message set associated with coding index $g_i$ is denoted by $\{1,\dots,2^{Nr_{g_i}}\}$.

We introduce the \textit{weighted likelihood} for channel output $Y^{(N)}$ and channel input vector $\mbox{\boldmath $X$}^{(N)} = \left[X_{1}^{(N)},X_{2}^{(N)}\right]^{\top}$ as
	\begin{equation}
		\ell_{\mbox{\scriptsize{\boldmath $g$}}}\left(\mbox{\boldmath $X$}^{(N)}, Y^{(N)}\right) \triangleq P\left(Y^{(N)}\middle \vert \mbox{\boldmath $X$}^{(N)}\right)P(\mbox{\boldmath $g$})2^{-N(r_{g_1}+r_{g_2})}.\label{WeightedLikelihood}
	\end{equation}

Given $Y^{(N)}$, decoder $D_{12}$ constructs the following constraint sets:
	\begin{eqnarray}
		&& \mathcal{S}_{\{1\}}^{(D_{12})} = \Big\{ (\mbox{\boldmath $w$},\mbox{\boldmath $g$}) : \mbox{\boldmath $g$ }  \in  \mathcal{R}_{12}, \tilde{\mbox{\boldmath $g$ }}\! \in\widehat{\mathcal{R}}_{12} \cup \mathcal{R}_{12c}, g_1 = \tilde{g}_{1}, \nonumber \\
		&&  \ \ \ \ \ \ \ \ \ \ \ell_{\mbox{\scriptsize \boldmath $g$}}\left(\mbox{\boldmath $X$}^{(N)},Y^{(N)}\right)\! > \! \gamma_{\{1\}}\left(\mbox{\boldmath $g$},\tilde{\mbox{\boldmath $g$}},\mbox{\boldmath $X$}^{(N)}, Y^{(N)}\right) \Big\},\nonumber
	\end{eqnarray}
	\begin{eqnarray}	
		&& \mathcal{S}_{\{2\}}^{(D_{12})} = \Big\{ (\mbox{\boldmath $w$},\mbox{\boldmath $g$}) : \mbox{\boldmath $g$ } \in \mathcal{R}_{12}, \tilde{\mbox{\boldmath $g$ }}\! \in \!\widehat{\mathcal{R}}_{12} \cup \mathcal{R}_{12c}, g_2 = \tilde{g}_{2},\nonumber \\
		&&\ \ \ \ \ \ \ \ \ \ \ell_{\mbox{\scriptsize \boldmath $g$}}\left(\mbox{\boldmath $X$}^{(N)},Y^{(N)}\right)\! > \! \gamma_{\{2\}}\left(\mbox{\boldmath $g$},\tilde{\mbox{\boldmath $g$}}, \mbox{\boldmath $X$}^{(N)},Y^{(N)}\right) \Big\},\nonumber
	\end{eqnarray}
	\begin{eqnarray}
		 && \mathcal{S}_{\{ \}}^{(D_{12})} = \Big\{ (\mbox{\boldmath $w$},\mbox{\boldmath $g$}) : \mbox{\boldmath $g$} \! \in \! \mathcal{R}_{12}, \tilde{\mbox{\boldmath $g$ }} \! \in \! \widehat{\mathcal{R}}_{12} \cup \mathcal{R}_{12c},\nonumber \\
		 &&\ \ \ \ \ \ \ \ \ \  \ell_{\mbox{\scriptsize \boldmath $g$}}\left(\mbox{\boldmath $X$}^{(N)},Y^{(N)}\right)\! > \! \gamma_{\{ \}}\left(\mbox{\boldmath $g$},\tilde{\mbox{\boldmath $g$}},\mbox{\boldmath $X$}^{(N)}, Y^{(N)}\right) \Big\}, \nonumber
	\end{eqnarray}

where $\mathcal{S}_{\{1\}}^{(D_{12})}$ and $\mathcal{S}_{\{2\}}^{(D_{12})}$ are the sets of message-code index pairs $(w_1,g_1)$ and $(w_2,g_2)$ satisfying the respective threshold tests for individual decoding, while $\mathcal{S}_{\{\}}$ corresponds to pairs $(\mbox{\boldmath $w$},\mbox{\boldmath $g$})$ satisfying the joint decoding criterion. The optimal thresholds $\gamma_{\{1\}}^{*}$, $\gamma_{\{2\}}^{*}$, and $\gamma_{\{\}}^{*}$ are chosen to minimize the bound on GEP. Although these thresholds depend on $Y^{(N)}$, $\mbox{\boldmath $X$}^{(N)}$, and $\tilde{\mbox{\boldmath $g$}}\in \mathcal{R}_{c}$, we suppress these dependencies for notational simplicity.

After constructing the constraint sets, decoder $D_{12}$ computes their intersection
		\begin{equation}
			\mathcal{S}_{I}^{(D_{12})} \triangleq \mathcal{S}_{\{1\}}^{(D_{12})} \cap \mathcal{S}_{\{2\}}^{(D_{12})} \cap \mathcal{S}_{\{ \}}^{(D_{12})},
		\end{equation}
and decodes by maximizing the weighted likelihood over $(\mbox{\boldmath $w$},\mbox{\boldmath $g$})\in\mathcal{S}_{I}^{(D_{12})}$ if $\mathcal{S}_{I}^{(D_{12})} \neq \emptyset$, i.e.,
		\begin{eqnarray}
			(\hat{\mbox{\boldmath $w$}}, \hat{\mbox{\boldmath $g$}}) = \arg\max_{(\mbox{\scriptsize \boldmath $w$}, \mbox{\scriptsize \boldmath $g$})\in \mathcal{S}_{I}^{(D_{12})}} \ell_{\mbox{\scriptsize \boldmath $g$}}\left(\mbox{\boldmath $X$}^{(N)},Y^{(N)}\right).
		\end{eqnarray}
Otherwise, it declares a collision.

For a given codebook, error probabilities are categorized into three types: \textit{incorrect decoding}, \textit{collision}, and \textit{miss-detection}. These are analyzed under the following scenarios:
		 \begin{enumerate}
		 	\item \textit{Case $\{\}$}: both users are decoded erroneously,
		 	\item \textit{Case $\{1\}$}: user $1$ is decoded successfully while user $2$ is decoded erroneously,
		 	\item \textit{Case $\{2\}$}: user $1$ is decoded erroneously while user $2$ is decoded successfully.
		 \end{enumerate}

For each case, we derive probability upper bounds, which together yield the final bound on $GEP_{12}$. Let $(\mbox{\boldmath $w$},\mbox{\boldmath $g$})$ be the actual message-code index pairs. The relevant definitions depend on whether $\mbox{\boldmath $g$}\in\mathcal{R}_{12}$ or $\mbox{\boldmath $g$}\in\widehat{\mathcal{R}}_{12}\cup \mathcal{R}_{12c}$. If $\mbox{\boldmath $g$}\in\mathcal{R}_{12}$, we consider incorrect decoding and collision probabilities, and if $\mbox{\boldmath $g$}\in\widehat{\mathcal{R}}_{12}\cup \mathcal{R}_{12c}$, we consider miss-detection probability. The analysis proceeds as follows.

		  \subsection*{Case $\{\}$ : Both Users Decoded Erroneously}
		  \begin{itemize}
		  	\item \textit{Incorrect Decoding Probability}: For $\mbox{\boldmath $g$}\in \mathcal{R}_{12}$, the probability that both users are incorrectly decoded is defined as
		  	\begin{eqnarray*}
		 	&& {P}_{i_{\{\}}}^{(D_{12})}(\mbox{\boldmath $g$}) \triangleq \mathbb{P}\Big[ \exists (\tilde{\mbox{\boldmath $w$}},\tilde{\mbox{\boldmath $g$}}) \neq (\mbox{\boldmath $w$},\mbox{\boldmath $g$}), \exists \tilde{\mbox{\boldmath $g$}}\in\mathcal{R}_{12} \  \mathrm{s.t.}\ \\
		 	&& \ \ \ \ \ \ \ \ \ \ \ \ \ \ell_{\tilde{\mbox{\scriptsize \boldmath $g$}}} \left(\tilde{\mbox{\boldmath $X$}}^{(N)},Y^{(N)} \right) \geq \ell_{\mbox{\scriptsize \boldmath $g$}}\left(\mbox{\boldmath $X$}^{(N)},Y^{(N)}\right) \Big].
		 \end{eqnarray*}
		 \item \textit{Collision Probability:} For $\mbox{\boldmath $g$}\in \mathcal{R}_{12}$, the collision probability is given by
		 \begin{eqnarray*} \label{CollisionEmpty}
		 	 &&\!\!\!\!\! {P}_{{c}_{\{\}}}^{(D_{12})}(\mbox{\boldmath $g$}) \triangleq \mathbb{P}\Big[\exists \tilde{\mbox {\boldmath $g$}} \in\widehat{\mathcal{R}}_{12}\cup \mathcal{R}_{12c}\ \mathrm{s.t.} \ \\
		 	 && \ \ \ \ \ \ \ \ \ \ \ \ \ \ \ \ \ \ \ \ \ \ \ \ \ \ \ell_{\mbox{\scriptsize \boldmath $g$}}\left(\mbox{\boldmath $X$}^{(N)},Y^{(N)}\right) \leq \gamma_{\{ \}}(\tilde{\mbox {\boldmath $g$}}) \Big].
		 \end{eqnarray*}
		 \item \textit{Miss-Detection Probability:} For $\mbox{\boldmath $g$}\in \widehat{\mathcal{R}}_{12} \cup \mathcal{R}_{12c
		 }$, the miss-detection probability is defined as
		 \begin{eqnarray*} \label{MissDetEmpty}
		 	  &&\!\!\! P_{{m}_{\{\}}}^{(D_{12})}(\mbox{\boldmath $g$})  \triangleq \mathbb{P}\Big[\exists (\tilde{\mbox{\boldmath $w$}},\tilde{\mbox{\boldmath $g$}}) \neq (\mbox{\boldmath $w$},\mbox{\boldmath $g$}), \tilde{\mbox{\boldmath $g$}} \in \mathcal{R}_{12} \ \mathrm{s.t.} \ \\
		 	  && \ \ \ \ \ \ \ \ \ \ \ \ \ \ \ \ \ \ \ \ \ \ \ \ \ \ \ \ell_{\tilde{\mbox{\scriptsize \boldmath $g$}}} \left(\tilde{\mbox{\boldmath $X$}}^{(N)}, Y^{(N)}\right) > \gamma_{\{ \}}(\tilde{\mbox {\boldmath $g$}})\Big],
		 \end{eqnarray*}
		 	 where $\tilde{\mbox{\boldmath $X$}}^{(N)}=\left [\tilde{X}_{1}^{(N)}, \tilde{X}_{2}^{(N)}\right]^{\top}$.
		  \end{itemize}
		  \subsection*{Case $\{1\}$: User 1 Correct, User 2 Erroneous}
		  \begin{itemize}
		  	\item \textit{Incorrect Decoding Probability:} For $\mbox{\boldmath $g$}\in \mathcal{R}_{12}$, the probability of incorrect decoding for user $2$ given successful decoding of user $1$ is
		  \begin{eqnarray*}
		  	&& \!\!\!\!\! {P}_{i_{\{1\}}}^{(D_{12})}(\mbox{\boldmath $g$})\triangleq \\
		  	&& \!\!\!\!\! \mathbb{P}\Big[ \exists (\tilde{w}_{2},\tilde{g}_{2}) \neq (w_2,g_2), \exists \tilde{\mbox{\boldmath $g$}}\in \mathcal{R}_{12}, (\tilde{w}_{1},\tilde{g}_{1})=(w_1,g_1) \ \\
		    &&  \ \ \ \ \ \ \ \mathrm{s.t.} \  \ell_{\tilde{\mbox{\scriptsize \boldmath $g$}}} \left(\tilde{\mbox{\boldmath $X$}}^{(N)},Y^{(N)} \right)\geq \ell_{\mbox{\scriptsize \boldmath $g$}} \left(\mbox{\boldmath $X$}^{(N)},Y^{(N)} \right)\Big].
		 \end{eqnarray*}
		 
		 \item \textit{Collision Probability:} For $\mbox{\boldmath $g$}\in \mathcal{R}_{12}$, the collision probability is
		  \begin{eqnarray*}
		  && \!\!\!\!\!\! {P}_{{c}_{\{1\}}}^{(D_{12})}(\mbox{\boldmath $g$}) \triangleq \\
		  && \!\!\!\!\!\! \mathbb{P}\Big[ \exists\tilde{\mbox{\boldmath $g$}}\in\widehat{\mathcal{R}}_{12}\cup \mathcal{R}_{12c},\tilde{g}_{1}=g_1 \ \mathrm{s.t.} \ \\
		  &&\ \ \ \ \ \ \ \ \ \ \ \ \ \ \ \ \ \ \ \ \ \ \ \ \  \ell_{\mbox{\scriptsize \boldmath $g$}}\left(\mbox{\boldmath $X$}^{(N)},Y^{(N)}\right) \leq  \gamma_{\{1\}}(\tilde{\mbox{\boldmath $g$}}) \Big].
		 \end{eqnarray*}
		 \item \textit{Miss-Detection Probability:} For $\mbox{\boldmath $g$} \in \widehat{\mathcal{R}}_{12} \cup \mathcal{R}_{12c}$, the miss-detection probability becomes
		  \begin{eqnarray*}
		  &&   \!\!\!\!\! P_{{m}_{\{1\}}}^{(D_{12})}(\mbox{\boldmath $g$}) \triangleq \\
		  &&   \!\!\!\!\! \mathbb{P}\Big[ \exists(\tilde{w}_{2}, \tilde{g}_{2})\neq (w_2,g_2), \exists \tilde{\mbox{\boldmath $g$}}\in\mathcal{R},(\tilde{w}_{1},\tilde{g}_{1})\!= \! (w_1,g_1) \\
		  &&  \ \ \ \ \ \ \ \ \ \ \ \ \ \ \ \ \ \ \  \ \ \mathrm{s.t.} \ \ell_{\tilde{\mbox{\scriptsize \boldmath $g$}}}\!\left(\tilde{\mbox{\boldmath $X$}}^{(N)},Y^{(N)}\! \right) > \gamma_{\{1\}}(\tilde{\mbox{\boldmath $g$}})\Big].
		 \end{eqnarray*}
		  \end{itemize}
		\subsection*{Case $\{2\}$: User 2 Correct, User 1 Erroneous}
		 This case mirrors \textit{Case $\{1\}$} with indices $1$ and $2$ interchanged. Thus, for $\mbox{\boldmath $g$}\in \mathcal{R}_{12}$
		 \begin{eqnarray*}
			P_{i_{\{2\}}}^{(D_{12})}(\mbox{\boldmath $g$}) = P_{i_{\{1\}}}(\mbox{\boldmath $g$})^{(D_{12})}\Big \vert_{1 \leftrightarrow 2}, \quad P_{c_{\{2\}}}^{(D_{12})}(\mbox{\boldmath $g$})=P_{c_{\{1\}}}^{(D_{12})}(\mbox{\boldmath $g$})\Big \vert_{1 \leftrightarrow 2},
		 \end{eqnarray*}
		and for $\mbox{\boldmath $g$}\in\widehat{\mathcal{R}}_{12} \cup \mathcal{R}_{12c}$,
		\begin{eqnarray*}
			P_{m_{\{2\}}}^{(D_{12})}(\mbox{\boldmath $g$}) = P_{m_{\{1\}}}^{(D_{12})}(\mbox{\boldmath $g$}) \Big \vert_{1\leftrightarrow 2}.
		\end{eqnarray*}

Having defined these probabilities, the generalized error performance in (\ref{GEP12Def}) can be expressed as
	\begin{eqnarray}\label{GEPForJoint}
		&& \!\!\!\!\! GEP_{12} =\nonumber \\
		&& \!\!\!\!\!\!\!\!\!\!\!\!\!\! \sum_{\mathcal{A}\in \{\{\},\{1\},\{2\}\}}\! \Bigg( \sum_{\mbox{\scriptsize \boldmath $g$} \in \mathcal{R}_{12}}P(\mbox{\boldmath $g$})\left(P_{i_{\mathcal{A}}}^{(D_{12})}(\mbox{\boldmath $g$})+P_{{c}_{\mathcal{A}}}^{(D_{12})}(\mbox{\boldmath $g$})\right) \\
		 	&&  \ \ \ \ \ \ \ \ \ \ \ \ \ \ \ \ \ \ \ \ \ \ \ \ \ \ \ \ \ \ \ +\sum_{\mbox{\scriptsize \boldmath $g$} \in \widehat{\mathcal{R}}_{12} \cup \mathcal{R}_{12c}}P(\mbox{\boldmath $g$})P_{{m}_{\mathcal{A}}}^{(D_{12})}(\mbox{\boldmath $g$}) \Bigg).\nonumber
		 \end{eqnarray}
We bound GEP by analyzing each term in (\ref{GEPForJoint}). The detailed derivation is presented for Case $\{\}$; the other cases follow analogously, and we provide only the final results.

By averaging over the codebook ensemble and applying the random coding union (RCU) bound, we obtain
		 \begin{eqnarray}
		 	&& \sum_{\mbox{\scriptsize \boldmath $g$}\in \mathcal{R}_{12}}P(\mbox{\boldmath $g$}) P_{i_{\{\}}}^{(D_{12})}(\mbox{\boldmath $g$})\leq \nonumber\\
		 	&& \sum_{\mbox{\scriptsize \boldmath $g$}\in \mathcal{R}_{12}}P(\mbox{\boldmath $g$}) \mathbb{E}\Bigg[\min\Bigg\{1,\sum_{\tilde{\mbox{\scriptsize \boldmath $g$}}\in \mathcal{R}_{12}}2^{N(r_{\tilde{g}_1}+r_{\tilde{g}_2})} \times \nonumber \\
		 	&&\!\!\!\!\!\! \mathbb{P}\left[\ell_{\tilde{\mbox{\scriptsize \boldmath $g$}}}\left(\tilde{\mbox{\boldmath $X$}}^{(N)},Y^{(N)}\right)\! \geq \! \ell_{\mbox{\scriptsize \boldmath $g$}}\left(\mbox{\boldmath $X$}^{(N)},Y^{(N)}\right)\middle\vert \mbox{\boldmath $X$}^{(N)}, Y^{(N)}\right]\Bigg \}\Bigg]\nonumber\\
			&& = \sum_{\mbox{\scriptsize \boldmath $g$}\in \mathcal{R}_{12}}\sum_{\tilde{\mbox{\scriptsize \boldmath $g$}}\in \mathcal{R}_{12}} P(\mbox{\boldmath $g$})\mathbb{E}\Bigg[\min\Bigg\{\frac{1}{|\mathcal{R}_{12}|}, 2^{N(r_{\tilde{g}_1}+r_{\tilde{g}_2})} \times \label{RCUJointCaseEmpty} \\
		 	&&\!\!\!\!\!\! \mathbb{P}\left[\ell_{\tilde{\mbox{\scriptsize \boldmath $g$}}}\left(\tilde{\mbox{\boldmath $X$}}^{(N)},Y^{(N)}\right) \! \geq \! \ell_{\mbox{\scriptsize \boldmath $g$}}\left(\mbox{\boldmath $X$}^{(N)},Y^{(N)}\right)\middle\vert \mbox{\boldmath $X$}^{(N)}, Y^{(N)}\right] \Bigg\}\Bigg],\nonumber
		 \end{eqnarray}
		 where the expectation is taken with respect to $P\left(Y^{(N)}\middle \vert \mbox{\boldmath $X$}^{(N)}\right)P_{\mbox{\scriptsize \boldmath $g$}}\left(\mbox{\boldmath $X$}^{(N)}\right)$.

For collision and miss-detection terms, we apply the union bound combined with codebook averaging, yielding
		 \begin{eqnarray}
		 	&&  \sum_{\mbox{\scriptsize \boldmath $g$}\in \mathcal{R}_{12}}P(\mbox{\boldmath $g$})P_{c_{\{\}}}^{(D_{12})}(\mbox{\boldmath $g$})+ \sum_{\mbox{\scriptsize \boldmath $g$}\in \widehat{\mathcal{R}}_{12}\cup\mathcal{R}_{12c}}\!\! P(\mbox{\boldmath $g$})P_{{m}_{\{\}}}^{(D_{12})}(\mbox{\boldmath $g$})\nonumber\\
		 	&& \!\!\!\!\!\!\!\!\!\!\!\! \leq \sum_{\mbox{\scriptsize \boldmath $g$}\in \mathcal{R}_{12}}\sum_{\tilde{\mbox{\scriptsize \boldmath $g$}} \in \widehat{\mathcal{R}}_{12}\cup \mathcal{R}_{12c}} \!\! P(\mbox{\boldmath $g$})\mathbb{P}\left[ \ell_{\mbox{\scriptsize \boldmath $g$}}\left(\mbox{\boldmath $X$}^{(N)},Y^{(N)}\right) \leq \gamma_{\{ \}} \right]\label{UBCase0Joint}\\
		 	&& \!\!\!\!\!\!\!\!\!\!\!\!\!\!\!\!\!\!\! +\!\!\! \sum_{\mbox{\scriptsize \boldmath $g$}\in \mathcal{R}_{12c}} \sum_{\tilde{\mbox{\scriptsize \boldmath $g$}}\in \mathcal{R}_{12}}\!\!\! P(\mbox{\boldmath $g$})2^{N(r_{\tilde{g}_1}+r_{\tilde{g}_2})}\mathbb{P} \left[\ell_{\tilde{\mbox{\scriptsize \boldmath $g$}}}\left(\tilde{\mbox{\boldmath $X$}}^{(N)},Y^{(N)}\!\right) \! > \! \gamma_{\{ \}}\! \right] \label{BeforSwappingCaseEmpty} \\
		 	&& \!\!\!\!\!\!\!\!\!\!\!\!\!\! = \sum_{\mbox{\scriptsize \boldmath $g$}\in \mathcal{R}_{12}}\sum_{\tilde{\mbox{\scriptsize \boldmath $g$}} \in \mathcal{R}_{12c}}\Bigg(P(\mbox{\boldmath $g$}) \mathbb{P}\left[\ell_{\mbox{\scriptsize \boldmath $g$}}\left(\mbox{\boldmath $X$}^{(N)},Y^{(N)}\right)\leq \gamma_{\{ \}}\right]+ \label{BeforSwappingCaseEmpty2} \\
		 	&& P(\tilde{\mbox {\boldmath $g$}})2^{N(r_{g_1}+r_{g_2})}\mathbb{P}\left[\ell_{\mbox{\scriptsize \boldmath $g$}}\left(\tilde{\mbox{\boldmath $X$}}^{(N)},Y^{(N)}\right)>\gamma_{\{ \}}\right]\Bigg), \label{SwappingCaseEmpty}
		 \end{eqnarray}
		where the second term in (\ref{SwappingCaseEmpty}) follows from swapping the indexes in (\ref{BeforSwappingCaseEmpty}). Note that the triple $\left(\mbox{\boldmath $X$}^{(N)},\tilde{\mbox{\boldmath $X$}}^{(N)}, Y^{(N)}\right)$ in (\ref{UBCase0Joint}), (\ref{BeforSwappingCaseEmpty}), and the first term in (\ref{SwappingCaseEmpty}) has joint distribution
		 \begin{eqnarray*}
		 	&&  P\left(Y^{(N)},\mbox{\boldmath $X$}^{(N)}, \tilde{\mbox{\boldmath $X$}}^{(N)}\right) = \\
		 	&&  \ \ \ \ \ \ \ \ \ \ \ \ \ \ \ P\left(Y^{(N)}\middle \vert \mbox{\boldmath $X$}^{(N)}\right)P_{\mbox{\scriptsize \boldmath $g$}}\left(\mbox{\boldmath $X$}^{(N)}\right) P_{\tilde{\mbox{\scriptsize \boldmath $g$}}}\left(\tilde{\mbox{\boldmath $X$}}^{(N)}\right),
		 \end{eqnarray*}
        while the triple $\left(\mbox{\boldmath $X$}^{(N)},\tilde{\mbox{\boldmath $X$}}^{(N)}, Y^{(N)}\right)$ in the second term in (\ref{SwappingCaseEmpty}) has joint distribution
          \begin{eqnarray*}
		 	&& P\left(Y^{(N)},\mbox{\boldmath $X$}^{(N)}, \tilde{\mbox{\boldmath $X$}}^{(N)}\right) = \\
		 	&& \ \ \ \ \ \ \ \ \ \ \ \ \ \ \ P\left(Y^{(N)}\middle \vert \mbox{\boldmath $X$}^{(N)}\right)P_{\tilde{\mbox{\scriptsize \boldmath $g$}}}\left(\mbox{\boldmath $X$}^{(N)}\right) P_{\mbox{\scriptsize \boldmath $g$}}\left(\tilde{\mbox{\boldmath $X$}}^{(N)}\right).
		 \end{eqnarray*}

By \cite[Proposition 18.3]{IT2025}, we rewrite (\ref{SwappingCaseEmpty}) as the following
		 \begin{eqnarray}
		 	&& \!\!\!\!\!\!\!\! \sum_{\mbox{\scriptsize \boldmath $g$}\in \mathcal{R}_{12}}P(\mbox{\boldmath $g$})P_{c_{\{\}}}^{(D_{12})}(\mbox{\boldmath $g$})+ \sum_{\mbox{\scriptsize \boldmath $g$}\in \widehat{\mathcal{R}}_{12}\cup\mathcal{R}_{12c}}P(\mbox{\boldmath $g$})P_{{m}_{\{\}}}^{(D_{12})}(\mbox{\boldmath $g$})\leq \nonumber \\
		 	&& \!\!\!\!\!\!\!\! \sum_{\mbox{\scriptsize \boldmath $g$}\in \mathcal{R}_{12}} \sum_{\tilde{\mbox{\scriptsize\boldmath $g$}}\in \widehat{\mathcal{R}}_{12}\cup\mathcal{R}_{12c}} \mathbb{E}\Bigg[P(\mbox{\boldmath $g$})\mathbf{1}\!\left\{\ell_{\mbox{\scriptsize \boldmath $g$}} \left(\mbox{\boldmath $X$}^{(N)},Y^{(N)}\right) \leq  \gamma_{\{ \}}\! \right\}  +  \label{CollisionMissJointCaseEmpt} \\
		 	&& \!\!\!\!\!\!\!\!\!\!\!\!\!\!\!\! \frac{P_{\tilde{\mbox{\scriptsize \boldmath $g$}}}\left(Y^{(N)}\right) P(\tilde{\mbox{\boldmath $g$}}) P(\mbox{\boldmath $g$})}{\ell_{\mbox{\scriptsize\boldmath $g$}}\left(\mbox{\boldmath $X$}^{(N)},Y^{(N)}\right)}\mathbf{1} \left\{\ell_{\mbox{\scriptsize \boldmath $g$ }} \left(\mbox{\boldmath $X$}^{(N)},Y^{(N)} \right) > \gamma_{\{ \}}\right\}\Bigg],\label{IndicatorLessthanOne1}
		 \end{eqnarray}
		 where $\mathbf{1}\{.\}$ is the indicator function, and the triple $\left(\mbox{\boldmath $X$}^{(N)},\tilde{\mbox{\boldmath $X$}}^{(N)}, Y^{(N)}\right)$ in both (\ref{CollisionMissJointCaseEmpt}) and (\ref{IndicatorLessthanOne1})  has joint distribution
		 \begin{eqnarray*}
		 	&& P\left(Y^{(N)},\mbox{\boldmath $X$}^{(N)}, \tilde{\mbox{\boldmath $X$}}^{(N)}\right) = \\
		 	&& \ \ \ \ \ \ \ \ \ \ \ \ \ \ \ P\left(Y^{(N)}\middle \vert \mbox{\boldmath $X$}^{(N)}\right)P_{\mbox{\scriptsize \boldmath $g$}}\left(\mbox{\boldmath $X$}^{(N)}\right) P_{\tilde{\mbox{\scriptsize \boldmath $g$}}}\left(\tilde{\mbox{\boldmath $X$}}^{(N)}\right) .
		 \end{eqnarray*}
		 We then minimize (\ref{CollisionMissJointCaseEmpt}) and (\ref{IndicatorLessthanOne1}) by the optimal threshold \cite[Theorem 18.6]{IT2025} as follows
		 \begin{eqnarray}\label{optimalThreshold12}
		 	\gamma_{\{ \}}^{*}\left(\tilde{\mbox{\boldmath $g$}},Y^{(N)}\right) = P(\tilde{\mbox{\boldmath $g$}}) P_{\tilde{\mbox{\scriptsize \boldmath $g$}}}\left(Y^{(N)}\right).
		 \end{eqnarray}
		 Here, $P_{\tilde{\mbox{\scriptsize \boldmath $g$}}}\left(Y^{(N)}\right)$ denotes the marginal probability of $Y^{(N)}$ given that the input symbol vector sequence $\mbox{\boldmath $X$}^{(N)}$ is generated according to $P_{\tilde{\mbox{\scriptsize \boldmath $g$}}}\left(\mbox{\boldmath $X$}^{(N)}\right)$.
		
		 Substituting (\ref{optimalThreshold12}) into (\ref{CollisionMissJointCaseEmpt}) and (\ref{IndicatorLessthanOne1}), we get
		 \begin{eqnarray}
		 	&& \!\!\!\!\!\!\!\! \sum_{\mbox{\scriptsize \boldmath $g$}\in \mathcal{R}_{12}}P(\mbox{\boldmath $g$})P_{c_{\{\}}}^{(D_{12})}(\mbox{\boldmath $g$})+ \sum_{\mbox{\scriptsize \boldmath $g$}\in \widehat{\mathcal{R}}_{12}\cup\mathcal{R}_{12c}}P(\mbox{\boldmath $g$})P_{{m}_{\{\}}}^{(D_{12})}(\mbox{\boldmath $g$})\leq \nonumber  \\
		 	&& \!\!\!\!\!\!\!\!\!\! \sum_{\mbox{\scriptsize \boldmath $g$}\in \mathcal{R}_{12}}\! \sum_{\tilde{\mbox{\scriptsize \boldmath $g$}}\in \widehat{\mathcal{R}}_{12}\cup\mathcal{R}_{12c}} \!\!\!\!\!\!\!\!\! P(\mbox{\boldmath $g$})\mathbb{P} \!\! \left[\frac{P_{\tilde{\mbox{\scriptsize \boldmath $g$}}}\left(Y^{(N)}\right)P(\tilde{\mbox{\boldmath $g$}})}{P\left(Y^{(N)}\middle \vert \mbox{\boldmath $X$}^{(N)}\right)P(\mbox{\boldmath $g$})} \! \geq \! 2^{-N(r_{{g}_{1}}+r_{{g}_{2}})} \! \right]\! +  \nonumber \\
		 	&& \!\!\! \label{SubstituteCaseEmpty}	\\	
		 	&&\!\!\!\!\!\!\!\!\!\!\!\!\!\! \mathbb{E}\left[\frac{P_{\tilde{\mbox{\scriptsize \boldmath $g$}}}\left(Y^{(N)}\right)P(\tilde{\mbox{\boldmath $g$}})P(\mbox{\boldmath $g$})}{\ell_{\mbox{\scriptsize \boldmath $g$}}\left(\mbox{\boldmath $X$}^{(N)},Y^{(N)}\right)}\mathbf{1}\!\! \left\{\frac{P_{\tilde{\mbox{\scriptsize \boldmath $g$}}}\left(Y^{(N)}\right)P(\tilde{\mbox{\boldmath $g$}})P(\mbox{\boldmath $g$})}{\ell_{\mbox{\scriptsize \boldmath $g$}}\! \left(\mbox{\boldmath $X$}^{(N)},Y^{(N)} \right)}\! < \!P(\mbox{\boldmath $g$})\right\}\right] \nonumber \\
		 	&& \label{IndicatorLessthanOne2}\\
		 	&& \!\!\! \leq 2 \sum_{\mbox{\scriptsize \boldmath $g$}\in \mathcal{R}_{12}}\sum_{\tilde{\mbox{\scriptsize \boldmath $g$}}\in \widehat{\mathcal{R}}_{12} \cup \mathcal{R}_{12c}} \min \Bigg\{P(\mbox{\boldmath $g$}),P(\tilde{\mbox{\boldmath $g$}})2^{N(r_{g_1}+r_{g_2})}\times \nonumber \\
		 	&& \ \ \ \ \ \ \ \ \ \ \ \ \ \ \ \ \ \  \ \ \ \ \ \ \ \ \ \ \  \  \mathbb{E}\left[\frac{P_{\tilde{\mbox{\scriptsize \boldmath $g$}}}\left(Y^{(N)}\right)}{P\left(Y^{(N)}\middle \vert \mbox{\boldmath $X$}^{(N)}\right)}\right]\Bigg\},\label{MarkovJointCaseEmpty}
		 \end{eqnarray}
		 where the bound in (\ref{MarkovJointCaseEmpty}) follows from applying Markov's inequality to (\ref{SubstituteCaseEmpty}) along with this fact that probability is less than or equal to one, and leveraging the property that indicator function is less than or equal to one in (\ref{IndicatorLessthanOne2}).

Similarly, following the RCU bounding procedure \cite{IT2025},  \cite{Haim2018}, the incorrect decoding probability for \textit{Case $\{1\}$} is
	\begin{eqnarray}\label{RCUForCase1}
		&&\sum_{\mbox{\scriptsize \boldmath $g$}\in \mathcal{R}_{12}}P(\mbox{\boldmath $g$})P_{i_{\{1\}}}^{(D_{12})}(\mbox{\boldmath $g$})\leq \nonumber\\
		&&\sum_{\mbox{\scriptsize \boldmath $g$}\in \mathcal{R}_{12}}\sum_{\scriptsize\begin{array}{c}	\tilde{\mbox{\boldmath $g$}}\in \mathcal{R}_{12} \\ \tilde{g}_{1}=g_1\end{array}} P(\mbox{\boldmath $g$})\mathbb{E}\Bigg[\min \Bigg\{\frac{1}{\left \vert \mathcal{R}_{{12}_{g_1}}\right\vert}, 2^{Nr_{\tilde{g}_{2}}} \times \label{RCUJointCase1}\\
		&& \!\!\!\!\!\!\!\!\!\!\!\!\!\! \mathbb{P}\left[\ell_{\tilde{\mbox{\scriptsize \boldmath $g$}}}\!\left({X}_{1}^{(N)}, \tilde{X}_{2}^{(N)},Y^{(N)}\!\right)\! >\! \ell_{\mbox{\scriptsize \boldmath $g$}}\! \left(\mbox{\boldmath $X$}^{(N)},Y^{(N)}\! \right) \middle \vert \mbox{\boldmath $X$}^{(N)}\!\!, Y^{(N)}\! \right]\! \Bigg\}\! \Bigg],\nonumber
	\end{eqnarray}
where $R_{{12}_{g_1}}\triangleq \{\tilde{\mbox{\boldmath $g$}}\in \mathcal{R}_{12} : \tilde{g}_{1}=g_1\}$, and for \textit{Case $\{2\}$} is
	\begin{eqnarray}
		&&\sum_{\mbox{\scriptsize \boldmath $g$}\in \mathcal{R}_{12}}P(\mbox{\boldmath $g$})P_{i_{\{2\}}}^{(D_{12})}(\mbox{\boldmath $g$})\leq \nonumber \\
		&& \sum_{\mbox{\scriptsize \boldmath $g$}\in \mathcal{R}_{12}}\sum_{\scriptsize\begin{array}{c}\tilde{\mbox{\boldmath $g$}}\in \mathcal{R}_{12} \\ \tilde{g}_{2}=g_2\end{array}} P(\mbox{\boldmath $g$}) \mathbb{E}\Bigg[\min \Bigg\{\frac{1}{\left \vert \mathcal{R}_{{12}_{g_2}}\right\vert}, 2^{Nr_{\tilde{g}_{1}}}\times \label{RCUJointCase2}\\
		&& \!\!\!\!\!\!\!\!\!\!\!\! \mathbb{P}\left[\ell_{\tilde{\mbox{\scriptsize \boldmath $g$}}}\! \left(\tilde{X}_{1}^{(N)}\! , {X}_{2}^{(N)}\! ,Y^{(N)}\! \right)\! > \! \ell_{\mbox{\scriptsize \boldmath $g$}}\! \left(\mbox{\boldmath $X$}^{(N)},Y^{(N)}\! \right) \middle \vert \mbox{\boldmath $X$}^{(N)} \!\! , Y^{(N)}\right]\! \Bigg\}\! \Bigg],\nonumber
	\end{eqnarray}
where $\mathcal{R}_{{12}_{g_2}}\triangleq \{\tilde{\mbox{\boldmath $g$}}\in \mathcal{R}_{12} : \tilde{g}_{2}=g_2\}$.

To bound the collision and miss-detection probabilities for \textit{Case $\{1\}$} in (\ref{GEPForJoint}), we adopt the analytical framework developed in (\ref{SwappingCaseEmpty}), (\ref{CollisionMissJointCaseEmpt}), and (\ref{IndicatorLessthanOne1}). By replicating the methodology applied to \textit{Case $\{\}$}, we have
	\begin{eqnarray}
		&& \!\!\!\!\!\!\!\!\!\!\!\!\! \sum_{\mbox{\scriptsize \boldmath $g$}\in \mathcal{R}_{12}}P(\mbox{\boldmath $g$})P_{{c}_{\{1\}}}^{(D_{12})}(\mbox{\boldmath $g$})+ \sum_{\mbox{\scriptsize \boldmath $g$}\in \widehat{\mathcal{R}}_{12}\cup\mathcal{R}_{12c}}P(\mbox{\boldmath $g$})P_{{m}_{\{1\}}}^{(D_{12})}(\mbox{\boldmath $g$})\leq \nonumber \\
	 	&&  \!\!\!\!\!\!\!\!\!\!\!\!\! \sum_{\mbox{\scriptsize \boldmath $g$}\in \mathcal{R}_{12}}\!\!\!\! \sum_{\scriptsize \begin{array}{c}\tilde{\mbox{\boldmath $g$}} \in\widehat{\mathcal{R}}_{12}\cup \mathcal{R}_{12c} \\ \tilde{g}_{1}=g_1\end{array}} \!\!\!\!\!\!\!\!\! \Bigg( P(\mbox{\boldmath $g$}) \mathbb{E}\left[\mathbf{1}\left\{\ell_{\mbox{\scriptsize \boldmath $g$}}\left(\mbox{\boldmath $X$}^{(N)},Y^{(N)}\right)\! \leq \! \gamma_{\{1\}}\right\}\right] \label{CollisionMissCase1Part1}\\
	 	&& \!\!\!\!\!\!\!\! + 2^{Nr_{{g}_{2}}}P(\tilde{\mbox{\boldmath $g$}}) \mathbb{E}\left[ \mathbf{1}\left\{\ell_{{\mbox{\scriptsize \boldmath $g$}}}\left(X_{1}^{(N)}, \tilde{X}_{2}^{(N)},Y^{(N)}\right) > \gamma_{\{1\}}\right\}\right]\Bigg)\label{CollisionMissCase1Part2}\\
	 	&& \!\!\!\!\!\!\!\!\!\!\!\!\!  = \!\!\!\! \sum_{\mbox{\scriptsize \boldmath $g$}\in \mathcal{R}_{12}}\!\!\!\! \sum_{\scriptsize \begin{array}{c}\tilde{\mbox{\boldmath $g$}} \in \widehat{\mathcal{R}}_{12} \cup \mathcal{R}_{12c} \\ \tilde{g}_{1}=g_1\end{array}} \!\!\!\!\!\!\!\!\! \mathbb{E}\Bigg[P(\mbox{\boldmath $g$})\mathbf{1}\left\{\ell_{\mbox{\scriptsize \boldmath $g$}}\left(\mbox{\boldmath $X$}^{(N)},Y^{(N)}\right)\! \leq \! \gamma_{\{1\}}\right\}\! +  \label{BoundCollisionMissCase1} \\
	 	&& \!\!\!\!\!\!\!\!\!\!\! \frac{P_{\tilde{g}_{2}}\left(Y^{(N)}\middle \vert X_{1}^{(N)}\right)P(\tilde{\mbox{\boldmath $g$}})P(\mbox{\boldmath $g$})}{2^{Nr_{g_1}}\ell_{\mbox{\scriptsize \boldmath $g$}}\! \left(\mbox{\boldmath $X$}^{(N)},Y^{(N)}\!  \right)} \mathbf{1}\left\{\ell_{\mbox{\scriptsize \boldmath $g$}}\left(\mbox{\boldmath $X$}^{(N)},Y^{(N)}\right)\! > \! \gamma_{\{1\}}\right\} \! \Bigg],\nonumber
	\end{eqnarray}
where $P_{\tilde{g}_{2}}\left(Y^{(N)}\middle \vert X_{1}^{(N)}\right)$ denotes the conditional probability of $Y^{(N)}$ given that the input symbol sequence $X_{2}^{(N)}$ is generated according to $P_{\tilde{g}_2}\left(X_2^{(N)}\right)$, and the expectation in (\ref{CollisionMissCase1Part2}) is taken with respect to joint distribution
    \begin{eqnarray*}
		 &&  P\left(\mbox{\boldmath $X$}^{(N)}, \tilde{X}_{2}^{(N)}, Y^{(N)}\right) = \\
	 	 &&  \ \ \ \ \ \ \ \ \ \ \ \ \ \ \  P\left(Y^{(N)}\middle \vert \mbox{\boldmath $X$}^{(N)}\! \right)\! P_{\mbox{\scriptsize \boldmath $g$}}\left( \mbox{\boldmath $X$}^{(N)}\right)P_{\tilde{g}_2}\left(\tilde{X}_{2}^{(N)}\right).
	 \end{eqnarray*}

From (\ref{BoundCollisionMissCase1}), the optimal threshold for \textit{Case $\{1\}$} is derived as
	\begin{eqnarray}\label{OptimalThresholdCase1}
		\!\!\!\! \gamma_{\{1\}}^{*}\! \left(\tilde{\mbox{\boldmath $g$}},Y^{(N)}\! \right)=  2^{-Nr_{{g}_{1}}}	P_{\tilde{g}_{2}}\left(Y^{(N)}\middle \vert X_{1}^{(N)}\right)P(\tilde{\mbox{\boldmath $g$}}).
	\end{eqnarray}
Substituting (\ref{OptimalThresholdCase1}) into (\ref{BoundCollisionMissCase1}), we obtain
	\begin{eqnarray}
		&& \!\!\!\!\!\!\!\!\! \sum_{\mbox{\scriptsize \boldmath $g$}\in \mathcal{R}_{12}}P(\mbox{\boldmath $g$})P_{{c}_{\{1\}}}^{(D_{12})}(\mbox{\boldmath $g$})+ \sum_{\mbox{\scriptsize \boldmath $g$}\in \widehat{\mathcal{R}}_{12}\cup \mathcal{R}_{12c}}P(\mbox{\boldmath $g$})P_{{m}_{\{1\}}}^{(D_{12})}(\mbox{\boldmath $g$})\leq \nonumber\\
		&& \!\!\!\!\!\!\!\!\!  \sum_{\mbox{\scriptsize \boldmath $g$}\in \mathcal{R}_{12}}\!\!\! \sum_{\scriptsize \begin{array}{c}\tilde{\mbox{\boldmath $g$}} \in \widehat{\mathcal{R}}_{12}\cup \mathcal{R}_{12c} \\ \tilde{g}_{1}=g_1\end{array}} \!\!\!\!\!\!\!\!\!\!\!\!\! P(\mbox{\boldmath $g$}) \mathbb{P}\!\left[\frac{P_{\tilde{g}_{2}}\left(Y^{(N)} \middle \vert X_{1}^{(N)}\right)P(\tilde{\mbox{\boldmath $g$}})}{P\left(Y^{(N)}\middle \vert \mbox{\boldmath $X$}^{(N)}\right)P(\mbox{\boldmath $g$})}\! \geq \! 2^{-Nr_{g_2}}\! \right] \nonumber \\
		&& \label{CollisionMissCase1_Part1}\\
		&&  + \mathbb{E} \Bigg [\frac{P_{\tilde{g}_{2}} \left(Y^{(N)}\middle \vert X_{1}^{(N)}  \right)\! P(\tilde{\mbox{\boldmath $g$}})P(\mbox{\boldmath $g$})}{2^{Nr_{g_1}}\ell_{\mbox{\scriptsize \boldmath $g$}} \left(\mbox{\boldmath $X$}^{(N)} ,Y^{(N)}\! \right)} \times \nonumber \\
		&& \ \ \ \ \ \ \ \ \ \mathbf{1}\Bigg\{\frac{P_{\tilde{g}_{2}} \left(Y^{(N)}\middle \vert X_{1}^{(N)} \right)P(\tilde{\mbox{\boldmath $g$}})P(\mbox{\boldmath $g$})}{2^{Nr_{g_1}}\ell_{\mbox{\scriptsize \boldmath $g$}} \left(\mbox{\boldmath $X$}^{(N)},Y^{(N)}\! \right)}\! < \! P(\mbox{\boldmath $g$})  \Bigg\} \Bigg]\label{CollisionMissCase1_Part2}\\
		&& \!\!\!\! \leq 2 \sum_{\mbox{\scriptsize \boldmath $g$}\in \mathcal{R}_{12}} \!\!\! \sum_{\scriptsize \begin{array}{c}\tilde{\mbox{\boldmath $g$}} \in \widehat{\mathcal{R}}_{12}\cup \mathcal{R}_{12c} \\ \tilde{g}_{1}=g_1\end{array}} \!\!\!\!\!\! \min \! \Bigg\{ P(\mbox{\boldmath $g$}), P(\tilde{\mbox{\boldmath $g$}}) 2^{Nr_{{g}_{2}}} \times \label{BoundCollisionMissCase1Final} \\
		&& \ \ \ \ \ \ \ \ \ \ \ \ \ \ \ \ \ \ \ \ \ \ \ \ \ \ \ \ \ \ \ \  \mathbb{E} \Bigg[\frac{P_{\tilde{g}_{2}}\left(Y^{(N)}\middle \vert X_{1}^{(N)}\right)}{P\left(Y^{(N)}\middle \vert \mbox{\boldmath $X$}^{(N)}\right)} \Bigg] \Bigg\}, \nonumber
	\end{eqnarray}
	where (\ref{BoundCollisionMissCase1Final}) follows from applying Markov's inequality to (\ref{CollisionMissCase1_Part1}), the trivial upper bound one on probability, and bounding above the indicator function in (\ref{CollisionMissCase1_Part2}) by one.
	
Following the analogous derivation to (\ref{BoundCollisionMissCase1}), the optimal threshold for \textit{Case $\{2\}$} is given by
	\begin{eqnarray}\label{OptimalThresholdCase2}
		\gamma_{\{2\}}^{*}\! \left(\tilde{\mbox{\boldmath $g$}},Y^{(N)}\! \right) \! = \! 2^{-Nr_{{g}_{2}}}P_{\tilde{g}_{1}}\! \left(Y^{(N)}\middle \vert X_{2}^{(N)}\! \right)P(\tilde{\mbox{\boldmath $g$}}).
	\end{eqnarray}
Using (\ref{OptimalThresholdCase2}), the summation of terms in (\ref{GEPForJoint}) corresponding to collision and miss-detection probabilities for \textit{Case $\{2\}$} is bounded above as
	\begin{eqnarray}
		&& \sum_{\mbox{\scriptsize \boldmath $g$}\in \mathcal{R}_{12}}P(\mbox{\boldmath $g$})P_{{c}_{\{2\}}}^{(D_{12})}(\mbox{\boldmath $g$})+ \sum_{\mbox{\scriptsize \boldmath $g$}\in \widehat{\mathcal{R}}_{12}\cup \mathcal{R}_{12c}}P(\mbox{\boldmath $g$})P_{{m}_{\{2\}}}^{(D_{12})}(\mbox{\boldmath $g$})\leq \nonumber\\
		&& 2 \sum_{\mbox{\scriptsize \boldmath $g$}\in \mathcal{R}_{12}} \sum_{\scriptsize \begin{array}{c}\tilde{\mbox{\boldmath $g$}} \in \widehat{\mathcal{R}}_{12} \cup \mathcal{R}_{12c} \\ \tilde{g}_{2}=g_2\end{array}}\!\!\!\!\!\! \min\Bigg\{P(\mbox{\boldmath $g$}), P(\tilde{\mbox{\boldmath $g$}}) 2^{Nr_{{g}_{1}}} \times  \label{BoundCollisionMissCase2Final} \\
	    && \ \ \ \ \ \ \ \ \ \ \ \ \ \ \ \ \ \ \ \ \ \ \ \ \ \ \ \ \ \ \ \ \mathbb{E}\Bigg[\frac{P_{\tilde{g}_{1}}\left(Y^{(N)} \middle \vert X_{2}^{(N)}\right)}{P\left(Y^{(N)}\middle \vert \mbox{\boldmath $X$}^{(N)}\! \right)}\!  \Bigg]\! \Bigg\},\nonumber
	\end{eqnarray}
where $P_{\tilde{g}_{1}}\! \left(Y^{(N)}\middle \vert X_{2}^{(N)}\right)$ denotes the conditional probability of $Y^{(N)}$ given that the input symbol sequence $X_{1}^{(N)}$ is generated according to $P_{\tilde{g}_1}\left(X_1^{(N)}\right)$, and the expectation is taken with respect to joint distribution
    \begin{eqnarray*}
		 &&  P\left(\mbox{\boldmath $X$}^{(N)}, \tilde{X}_{1}^{(N)}, Y^{(N)}\right) = \\
   	 && \ \ \ \ \ \ \ \ \ \ \ \ \ \ \ P\left(Y^{(N)}\middle \vert \mbox{\boldmath $X$}^{(N)}\! \right)\! P_{\mbox{\scriptsize \boldmath $g$}}\left( \mbox{\boldmath $X$}^{(N)}\right)P_{\tilde{g}_1}\left(\tilde{X}_{1}^{(N)}\right).
	 \end{eqnarray*}

Substituting the upper bounds from (\ref{RCUJointCaseEmpty}), (\ref{MarkovJointCaseEmpty}), (\ref{RCUJointCase1}), (\ref{RCUJointCase2}), (\ref{BoundCollisionMissCase1Final}), and (\ref{BoundCollisionMissCase2Final}) into (\ref{GEPForJoint}) completes the proof.
\end{proof}

\section{Proof of Theorem \ref{GEP1}} \label{AppGEP1}

\begin{proof}
The proof follows the same reasoning as Theorem~\ref{GEP12}, but specialized to the single-user setting of decoder $D_1$. Here, decoding is performed by maximizing the weighted likelihood over the intersection of the decision sets $\mathcal{S}_{\{ \}}^{(D_{1})}$ and $\mathcal{S}^{(D_{1})}_{\{1\}}$. Decoder $D_1$ selects the candidate with the highest weighted likelihood if the intersection is non-empty, and declares a collision otherwise:
\begin{eqnarray}
		(\hat{w}_{1},\hat{\mbox{\boldmath $g$}}) =  \underset{
  \scriptsize
  \begin{array}{@{}c@{}}
    ({w}_{1}, \mbox{\boldmath $g$}) \in  \mathcal{S}_{\{ \}}^{(D_1)}\cap\mathcal{S}_{\{1\}}^{(D_{1})} \\
    \mathcal{S}_{\{ \}}^{(D_1)}\cap\mathcal{S}_{\{1\}}^{(D_{1})}  \neq \emptyset
  \end{array} }
{\arg\max} \ell_{\mbox{\scriptsize \boldmath $g$}}\left(X_{1}^{(N)}, Y^{(N)}\right),
	\end{eqnarray}
	 where the weighted likelihood is
	 \begin{eqnarray}
	 	\ell_{\mbox{\scriptsize \boldmath $g$}}\left(X_{1}^{(N)}, Y^{(N)}\right) = P\left(Y^{(N)} \middle \vert X_1^{(N)}, g_2\right)P(\mbox{\boldmath $g$})2^{-Nr_{g_1}}, \nonumber
	 \end{eqnarray}
	 and the decision sets for sub-decoder $D_1$ are
	\begin{eqnarray}
		&& \!\!\!\!\!\!\!\!\!\!\! \mathcal{S}_{\{ \}}^{(D_{1})} = \Big\{ (w_1, \mbox{\boldmath $g$}): \mbox{\boldmath $g$}\in \mathcal{R}_{1}, \tilde{\mbox{\boldmath $g$}}\in \widehat{\mathcal{R}}_{1}\cup \mathcal{R}_{1c},\nonumber \\
		&& \ \ \  \ \ \ \ \ \ \ \ \ \ell_{\mbox{\scriptsize \boldmath $g$}}\left(X_{1}^{(N)},Y^{(N)}\right) > \gamma_{\{ \}}\left(\mbox{\boldmath $g$},\tilde{\mbox{\boldmath $g$}},Y^{(N)}\right)\Big\},\nonumber
	\end{eqnarray}
	\begin{eqnarray}
		&& \!\!\!\!\!\!\!\!\! \mathcal{S}^{(D_1)}_{\{1\}} = \Big\{(w_1,\mbox{\boldmath $g$}) : \mbox{\boldmath $g$}\in \mathcal{R}_{1}, \tilde{\mbox{\boldmath $g$}}\in \widehat{\mathcal{R}}_{1}\cup \mathcal{R}_{1c}, \tilde{g}_{1}=g_1, \nonumber \\
		&&\ \ \ \ \ \ \ \ \ell_{\mbox{\scriptsize \boldmath $g$}}\left(X_1^{(N)},Y^{(N)}\right)> \gamma_{\{1\}} \left(\tilde{\mbox{\boldmath $g$}},X_1^{(N)},Y^{(N)}\right)\Big\},\nonumber
	\end{eqnarray}
	with thresholds $\gamma_{\{ \}}$ and $\gamma_{\{1\}}$.

For the actual message $w_1$ and augmented code index vector $\mbox{\boldmath $g$}\in \mathcal{R}_{1}$, we have the following cases.

In \textit{Case $\{\}$}, the incorrect decoding and collision probabilities are, respectively,
	\begin{eqnarray}
		&& \!\!\!\!\!\!\!\!\!\! P_{i_{\{\}}}^{(D_1)}(\mbox{\boldmath $g$}) \triangleq \mathbb{P}\Big[\exists (\tilde{w}_{1},\tilde{g}_{1})\neq (w_1,g_1), \exists \tilde{\mbox{\boldmath $g$}}\in \mathcal{R}_{1} \ \mathrm{s.t.} \ \nonumber\\
		&& \ \ \ \ \  \ \ \ \ \ \  \ \ \ \ \ \ \ell_{\tilde{\mbox{\scriptsize \boldmath $g$}}} \left(\tilde{X}_{1}^{(N)}, Y^{(N)} \right) \geq  \ell_{\mbox{\scriptsize \boldmath $g$}}\left(X_{1}^{(N)}, Y^{(N)}\right) \Big],\nonumber
	\end{eqnarray}
	\begin{eqnarray}
		P_{c_{\{\}}}^{(D_1)}(\mbox{\boldmath $g$})\! \triangleq \! \mathbb{P}\Big[\exists \tilde{\mbox{\boldmath $g$}}\in \widehat{\mathcal{R}}_{1}\cup \mathcal{R}_{1c} \ \mathrm{s.t.} \ \ell_{\mbox{\scriptsize \boldmath $g$}}\! \left(X_{1}^{(N)}\! , Y^{(N)} \! \right)\! \leq \! \gamma_{\{ \}}\Big].\nonumber
	\end{eqnarray}
For $\mbox{\boldmath $g$} \in \widehat{\mathcal{R}}_{1}\cup \mathcal{R}_{1c}$, the miss-detection probability for \textit{Case $\{\}$} is
	\begin{eqnarray}
		 && \!\!\!\!\!\!\!\!\!\!\!\! P_{m_{\{\}}}^{(D_1)}(\mbox{\boldmath $g$}) \triangleq \mathbb{P}\Big[\exists (\tilde{w}_1,\tilde{g}_{1}) \neq (w_1,g_1), \exists \tilde{\mbox{\boldmath $g$}} \in \mathcal{R}_{1} \ \mathrm{s.t.} \ \nonumber \\
		 &&  \ \ \ \ \ \ \ \ \ \ \ \ \ \  \ \ \ \ \ \ \ \ \ \ \ \ \ \ \ \ \ \ \ \ \ \ell_{\tilde{\mbox{\scriptsize \boldmath $g$}}}\! \left(\tilde{X}_{1}^{(N)}, Y^{(N)}\! \right) \! > \! \gamma_{\{ \}} \Big].\nonumber
	\end{eqnarray}
	
In \textit{Case $\{1\}$}, when $\mbox{\boldmath $g$} \in \mathcal{R}_{1}$, the collision probability is defined as
\begin{eqnarray}
		 && P_{c_{\{1\}}}^{(D_1)}(\mbox{\boldmath $g$})\! \triangleq \nonumber \\
		 && \mathbb{P} \Big[ \exists \tilde{\mbox{\boldmath $g$}} \in \mathcal{R}_{1c}, \tilde{g}_{1}=g_1 \ \mathrm{s.t.} \ \ell_{\mbox{\scriptsize \boldmath $g$}} \left({X}_{1}^{(N)}, Y^{(N)} \right)  \leq  \gamma_{\{1\}} \Big], \nonumber
	\end{eqnarray}

and when $\mbox{\boldmath $g$} \in \mathcal{R}_{1c}$, the miss-detection probability is defined as
\begin{eqnarray}
		&& \!\! P_{m_{\{1\}}}^{(D_1)}(\mbox{\boldmath $g$}) \triangleq \mathbb{P}\Big[\exists \tilde{\mbox{\boldmath $g$}} \in \mathcal{R}_{1}, (\tilde{w}_{1},\tilde{g}_{1})=(w_1,g_1) \ \mathrm{s.t.} \ \nonumber \\
		&& \ \ \ \ \ \ \ \ \ \ \ \ \ \ \ \ \ \ \ \ \ \ \ \ \ \ \ \ \ \ \ \ \ \ \ \ell_{\tilde{\mbox{\scriptsize \boldmath $g$}}}\! \left(\tilde{X}_{1}^{(N)}, Y^{(N)}\! \right)\! >\!  \gamma_{\{1\}} \Big]. \nonumber	
\end{eqnarray}

By definition of $\mathrm{GEP}_{1}$ in (\ref{GEP1}), we obtain
\begin{eqnarray}
		&& \!\!\!\! {GEP}_{1} = \sum_{\mbox{\scriptsize \boldmath $g$} \in \mathcal{R}_{1}}P(\mbox{\boldmath $g$}) \left( P_{i_{\{\}}}^{(D_1)}(\mbox{\boldmath  $g$}) + P_{c_{\{\}}}^{(D_1)}(\mbox{\boldmath $g$})\right) \label{GEP1_in_proof}\\
		&& \!\!\!\! + \sum_{\mbox{\scriptsize \boldmath $g$}\in \widehat{\mathcal{R}}_{1}\cup \mathcal{R}_{1c}}P(\mbox{\boldmath $g$})P_{m_{\{\}}}^{(D_1)}(\mbox{\boldmath $g$})  \nonumber\\
		&&  \!\!\!\! + \sum_{\mbox{\scriptsize \boldmath $g$}\in \mathcal{R}_{1}}P(\mbox{\boldmath $g$}) P_{m_{\{1\}}}^{(D_1)}(\mbox{\boldmath $g$}) + \sum_{\mbox{\scriptsize \boldmath $g$}\in \mathcal{R}_{1c}}p(\mbox{\boldmath $g$})P_{c_{\{1\}}}^{(D_1)}(\mbox{\boldmath $g$}). \nonumber
\end{eqnarray}

Using the RCU bound, analogous to Theorem~\ref{GEP12}, we upper bound the incorrect decoding probability of decoder $D_1$ as
\begin{eqnarray}
	&& \sum_{\mbox{\scriptsize \boldmath $g$}\in \mathcal{R}_{1}}P(\mbox{\boldmath $g$})P_{i_{\{\}}}^{(D_1)}(\mbox{\boldmath $g$})\leq \nonumber\\
	&& \sum_{\mbox{\scriptsize \boldmath $g$}\in \mathcal{R}_{1}}\sum_{\tilde{\mbox{\scriptsize \boldmath $g$}}\in \mathcal{R}_{1}} P(\mbox{\boldmath $g$}) \mathbb{E}\Bigg[\min \Bigg\{\frac{1}{| \mathcal{R}_{1}|}, 2^{Nr_{\tilde{g}_{1}}} \times \label{IncDecThem3} \\
	&& \!\!\!\!\!\!\!\!\!\!\! \mathbb{P}\left[\ell_{\tilde{\mbox{\scriptsize \boldmath $g$}}}\!\left (\tilde{X}_{1}^{(N)}, Y^{(N)}\right)\! >\!  \ell_{\mbox{\scriptsize \boldmath $g$}}\! \left(X_{1}^{(N)},Y^{(N)}\right)\middle \vert X_{1}^{(N)}\! , Y^{(N)}\right]\!\! \Bigg\}\Bigg].\nonumber
\end{eqnarray}

To upper bound the combined collision and miss-detection probability for \textit{Case $\{\}$}, we adopt the same approach as in Theorem~\ref{GEP12}, applying index swapping and Markov’s inequality, which yields
\begin{eqnarray}
	&& \!\!\!\! \sum_{\mbox{\scriptsize \boldmath $g$}\in \mathcal{R}_{1}}P(\mbox{\boldmath $g$})P_{c_{\{\}}}^{(D_1)}(\mbox{\boldmath $g$})+ \!\!\! \sum_{\mbox{\scriptsize \boldmath $g$}\in \widehat{\mathcal{R}}_{1}\cup \mathcal{R}_{1c}}P(\mbox{\boldmath $g$})P_{m_{\{\}}}^{(D_1)}(\mbox{\boldmath $g$}) \leq \nonumber\\
	&& \!\!\!\! 2 \sum_{\mbox{\scriptsize \boldmath $g$}\in \mathcal{R}_{1}}\sum_{\tilde{\mbox{\scriptsize \boldmath $g$}}\in \widehat{\mathcal{R}}_{1}\cup \mathcal{R}_{1c}}\!\!\!\!\! \min \Bigg\{P(\mbox{\boldmath $g$}), P(\tilde{\mbox{\boldmath $g$}})2^{Nr_{g_1}}\times \nonumber\\
	&& \ \ \ \ \ \ \ \ \ \ \ \ \ \ \ \ \ \ \ \ \ \ \ \mathbb{E}\Bigg[\frac{P_{\tilde{g}_{1}}\left(Y^{(N)}\middle \vert g_2\right)}{P\left(Y^{(N)}\middle \vert X_{1}^{(N)},g_2\right)}\Bigg]\Bigg\}, \label{MissColThm3Empty}
\end{eqnarray}
	with optimal threshold
	\begin{eqnarray}
		\gamma_{\{ \}}^{*}\left(\tilde{\mbox{\boldmath $g$}},Y^{(N)}\right) = P(\tilde{\mbox{\boldmath $g$}})P_{\tilde{\mbox{\scriptsize \boldmath $g$}}}\left(Y^{(N)}\right),\nonumber
	\end{eqnarray}
where $\tilde{\mbox{\boldmath $g$}} = [\tilde{g}_{1},g_2]^{\top}$.

Finally, we bound the terms corresponding to \textit{Case $\{1\}$} in (\ref{GEP1_in_proof}) by

\begin{eqnarray}
		 && \!\!\!\!\!\!\!\!\!\! \sum_{\mbox{\scriptsize \boldmath $g$}\in \mathcal{R}_{1c}}P(\mbox{\boldmath $g$})P_{m_{\{1\}}}^{(D_1)}(\mbox{\boldmath $g$})+ \sum_{\mbox{\scriptsize \boldmath $g$}\in \mathcal{R}_{1}}P(\mbox{\boldmath $g$})P_{c_{\{1\}}}^{(D_1)}(\mbox{\boldmath $g$}) \leq \nonumber\\
		 && \!\!\!\!\!\!\!\!\!\! 2 \sum_{\mbox{\scriptsize \boldmath $g$}\in \mathcal{R}_{1}}\!\!\!\! \sum_{\scriptsize\begin{array}{c}\tilde{\mbox{\boldmath $g$}}\in \mathcal{R}_{1c} \\ \tilde{g}_{1}=g_1 \end{array}}\!\!\!\!\!\! \min \Bigg\{ P(\mbox{\boldmath $g$}), P(\tilde{\mbox{\boldmath $g$}}) \mathbb{E}\Bigg[\frac{P_{\tilde{\mbox{\scriptsize \boldmath $g$}}}\left(Y^{(N)}\middle \vert X_{1}^{(N)}\right)}{P\left(Y^{(N)}\middle \vert X_{1}^{(N)},g_2\right)}\Bigg]\Bigg\}, \nonumber \\
		 && \label{MissColThm3Case1}
\end{eqnarray}

with optimal threshold
\begin{eqnarray}
	&& \gamma_{\{1\}}^{*}\left(\tilde{\mbox{\boldmath $g$}},X_{1}^{(N)},Y^{(N)}\right) = \nonumber\\
	&&\ \ \ \ \ \ \ \ \ \ \ \ \  \ \ \ \ \  2^{-Nr_{g_1}}P(\tilde{\mbox{\boldmath $g$}})P_{\tilde{\mbox{\scriptsize \boldmath $g$}}}\left(Y^{(N)}\middle \vert X_{1}^{(N)}\right),\nonumber
\end{eqnarray}
where $\tilde{\mbox{\boldmath $g$}} = [g_1,\tilde{g}_{2}]^{\top}$.

Having upper bounded all terms in $\mathrm{GEP}_{1}$, substituting (\ref{IncDecThem3}), (\ref{MissColThm3Empty}), and (\ref{MissColThm3Case1}) into (\ref{GEP1_in_proof}) completes the proof.
\end{proof}

\section{Proof of Theorem \ref{Theorem1UserDecodingD}} \label{AppGEPD}
\begin{proof}
For a fixed $S \subseteq D$, let
\begin{equation*}
\mbox{\boldmath $X$}_{S}^{(N)} = \left [X_{i}^{(N)}(w_{ij}, g_{ij})\right]_{[i,j]\in {S}}
\end{equation*}
denote the list of codewords corresponding to $\mbox{\boldmath $w$}_{S} = [w_{ij}]_{[i,j]\in S}$, and let $\mbox{\boldmath $g$}_{S} = \left[g_{ij}\right]_{[i,j]\in S}$ be the associated code index list. The elements of $\mbox{\boldmath $w$}_{S}$ and $\mbox{\boldmath $g$}_{S}$ are listed in the same order of their subscripts. For simplicity, we abbreviate $X_{i}^{(N)}(w_{ij}, g_{ij})$ as $X_{ij}^{(N)}$ throughout the proof.

Given channel output $Y^{(NL+T_2)}$ and a fixed $D \subseteq U$, we define the weighted likelihood as

\begin{eqnarray}
		&& \!\!\!\!\!\!\! \ell_{\mbox{\scriptsize \boldmath $g$}_{U}}\left(\mbox{\boldmath $X$}_{D}^{(N)},Y^{(NL+T_2)}\right) \triangleq \nonumber \\
		&& P \left(Y^{(NL+T_2)}  \middle \vert \mbox{\boldmath $X$}_{D}^{(N)},\mbox{\boldmath $g$}_{U \backslash D} \right) P(\mbox{\boldmath $g$}_{U})2^{-N\sum_{[i,j]\in D}r_{g_{ij}}}. \nonumber
\end{eqnarray}
To decode $\mbox{\boldmath $w$}_{D}$, sub-decoder $D_D$ establishes the following constraint set for all $S \subseteq D$:

\begin{eqnarray}
		&& \!\!\!\!\!\!\!\! \mathcal{S}_{S} = \Big\{ (\mbox{\boldmath $w$}_{D},\mbox{\boldmath $g$}_{U}) : \mbox{\boldmath $g$}_{U} \in \mathcal{R}_{D}, \tilde{\mbox{\boldmath $g$}}_{U}\in \widehat{\mathcal{R}}_{D}\cup \mathcal{R}_{Dc}, \tilde{\mbox{\boldmath $g$}}_{S}=\mbox{\boldmath $g$}_{S},\nonumber\\  &&\ \ \ \ \ \ \ \ \ \ \ \ \ \ \ \ \ \ \ \ \ \ \ \ \ \ \ \ \ell_{\mbox{\scriptsize \boldmath $g$}_{U}}\left(\mbox{\boldmath $X$}_{D}^{(N)},Y^{(NL+T_2)}\right) > \gamma_{S}\Big\},\nonumber
\end{eqnarray}
where $\gamma_S$ is a pre-determined threshold whose optimal value and parameter dependencies are derived in the course of this proof.

The next decoding step is performed by maximizing the weighted likelihood over the intersection of all constraint sets:
\begin{eqnarray}
(\hat{\mbox{\boldmath $w$}}_{D}, \hat{\mbox{\boldmath $g$}}_{U}) = \underset{
  \scriptsize
  \begin{array}{@{}c@{}}
    (\mbox{\boldmath $w$}_{D}, \mbox{\boldmath $g$}_{U}) \in \mathcal{S}_{I}^{(D_{D})} \\
    \mathcal{S}_{I}^{(D_{D})}  \neq \emptyset
  \end{array} }
{\arg\max} \ell_{\mbox{\scriptsize \boldmath $g$}_{U}}\left(\mbox{\boldmath $X$}_{D}^{(N)}, Y^{(NL+T_2)}\right),
\end{eqnarray}
with
	\begin{eqnarray}
		\mathcal{S}_{I}^{(D_D)} = \bigcap_{S \subseteq D} \mathcal{S}_{S}.
	\end{eqnarray}

\subsection*{Error probabilities}
For a fixed $S \subset D$, we define incorrect decoding, collision, and miss-detection probabilities for the case where the subset $S$ of messages is decoded successfully while the remaining $D \setminus S$ are decoded erroneously. Analogous to Theorem~\ref{GEP12}:
	\begin{itemize}
	\item for $\mbox{\boldmath $g$}_{U}\in \mathcal{R}_{D}$, we define the \emph{incorrect decoding probability} as		
		\begin{eqnarray}
			&& \!\!\!\!\!\!\!\!\!\!\!\!\!\! P_{iS}(\mbox{\boldmath $g$}_{U}) \triangleq \mathbb{P}\Big[ \exists \tilde{\mbox{\boldmath  $w$}}_{D \backslash S}, \exists \tilde{\mbox{\boldmath $g$}}_{U}\in \mathcal{R}_{D}, (\tilde{\mbox{\boldmath $w$}}_{S}, \tilde{\mbox{\boldmath $g$}}_{S}) = (\mbox{\boldmath $w$}_{S}, \mbox{\boldmath $g$}_{S}) \nonumber \\ 
			&& \!\!\!\!\!\!\!\!\!\!\!\!\!\! \mathrm{s.t.} \ \ell_{\tilde{\mbox{\scriptsize \boldmath $g$}}_{U}}\left(\tilde{\mbox{\boldmath $X$}}_{D}^{(N)}, Y^{(NL+T_2)}\right)\! > \! \ell_{\mbox{\scriptsize \boldmath $g$}_{U}}\left(\mbox{\boldmath $X$}_{D}^{(N)}, Y^{(NL+T_2)}\right)\Big],\nonumber
		\end{eqnarray}
	\item for $\mbox{\boldmath $g$}_{U}\in \mathcal{R}_{D}$, the \emph{collision probability} as		
		\begin{eqnarray}
			&& \!\!\!\!\!\!\!\!\!\!\!\!\!\! P_{cS}(\mbox{\boldmath $g$}_{U}) \triangleq \mathbb{P}\Big[\exists \tilde{\mbox{\boldmath $g$}}_{U}\in \widehat{\mathcal{R}}_{D}\cup \mathcal{R}_{Dc}, \tilde{\mbox{\boldmath $g$}}_{S}=\mbox{\boldmath $g$}_{S} \nonumber \\
			&&\ \ \ \ \ \ \ \ \ \ \ \ \ \ \ \ \ \ \mathrm{s.t.}\ \ell_{\mbox{\scriptsize \boldmath $g$}_{U}}\left(\mbox{\boldmath $X$}_{D}^{(N)}, Y^{(NL+T_2)}\right)\leq \gamma_S \Big],\nonumber
		\end{eqnarray}
	\item for $\mbox{\boldmath $g$}_{U}\in \widehat{\mathcal{R}}_{D}\cup \mathcal{R}_{Dc}$, the \emph{miss-detection probability} as

		\begin{eqnarray}
			&& \!\!\!\!\!\!\!\!\!\!\!\!\!\!\! P_{mS}(\mbox{\boldmath $g$}_{U}) \triangleq \mathbb{P}\Big[\exists \tilde{\mbox{\boldmath $w$}}_{D \backslash S}, \exists \tilde{\mbox{\boldmath $g$}}_{U}\in \mathcal{R}_{D}, (\tilde{\mbox{\boldmath $w$}}_{S}, \tilde{\mbox{\boldmath $g$}}_{S}) = (\mbox{\boldmath $w$}_{S}, \mbox{\boldmath $g$}_{S}) \nonumber \\
			&&\ \ \ \ \ \ \ \ \ \ \ \ \ \ \ \ \ \  \mathrm{s.t.} \ \ell_{\tilde{\mbox{\scriptsize \boldmath $g$}}_{U}}\left(\tilde{\mbox{\boldmath $X$}}_{D}^{(N)}, Y^{(NL+T_2)}\right) > \gamma_S \Big].\nonumber
		\end{eqnarray}
	\end{itemize}

For $S=D$, we also define:
\begin{itemize}
    \item the \emph{collision probability} when $\mbox{\boldmath $g$}_{U}\in \mathcal{R}_{D}$,
	\begin{eqnarray}
		&& \!\!\!\!\!\!\!\!\!\!\!\!\!\!\! P_{cD}(\mbox{\boldmath $g$}_{U}) \triangleq \mathbb{P}\Big[ \exists \tilde{\mbox{\boldmath $g$}}_{U} \in \mathcal{R}_{Dc}, \tilde{\mbox{\boldmath $g$}}_{D}=\mbox{\boldmath $g$}_{D} \ \mathrm{s.t.} \nonumber \\
		&& \ \ \ \ \ \ \ \ \ \ \ \ \ \ \ \ \ \ \ \ \ \ell_{\mbox{\scriptsize \boldmath $g$}_{U}}\left(\mbox{\boldmath $X$}_{D}^{(N)},Y^{(NL+T_2)}\right) \leq \gamma_D\Big].\nonumber
	\end{eqnarray}
    \item and the \emph{miss-detection probability} when $\mbox{\boldmath $g$}_{U}\in \mathcal{R}_{Dc}$,
	\begin{eqnarray}
		&& \!\!\!\!\!\!\!\!\!\!\!\!\!\!\! P_{mD}(\mbox{\boldmath $g$}_{U}) \triangleq \mathbb{P}\Big[ \exists \tilde{\mbox{\boldmath $g$}}_{U} \in \mathcal{R}_{D}, (\tilde{\mbox{\boldmath $w$}}_{D},\tilde{\mbox{\boldmath $g$}}_{D}) = (\mbox{\boldmath $w$}_{D},\mbox{\boldmath $g$}_{D}) \ \mathrm{s.t.} \nonumber \\
		&& \ \ \ \ \ \ \ \ \ \ \ \ \ \ \ \ \ \ \ \ \ \ell_{\tilde{\mbox{\scriptsize \boldmath $g$}}_{U}}\left(\tilde{\mbox{\boldmath $X$}}_{D}^{(N)}, Y^{(NL+T_2)}\right)> \gamma_D\Big],\nonumber
	\end{eqnarray}
\end{itemize}

By (\ref{ExtendedGEPM}), the generalized error performance for decoder $D_D$ can be expressed as
\begin{eqnarray}
		&& {GEP}_{D} = \sum_{S \subset D}\Bigg(\sum_{\mbox{\scriptsize \boldmath $g$}_{U}\in \mathcal{R}_{D}}P(\mbox{\boldmath $g$}_{U})\Big(P_{iS}(\mbox{\boldmath $g$}_{U})+P_{cS}(\mbox{\boldmath $g$}_{U})\Big) \nonumber \\
	    && \ \ \ \ \ \ \ \ \ \ \ \ \ \ \ \ \ \ \ \ \ \ \ \ \  + \sum_{\mbox{\scriptsize \boldmath $g$}_{U}\in \widehat{\mathcal{R}}_{D}\cup \mathcal{R}_{Dc}}P(\mbox{\boldmath $g$}_{U})P_{mS}(\mbox{\boldmath $g$}_{U})\Bigg)  \nonumber \\
	    && \!\!\!\!\!\!\!\!\! + \sum_{\mbox{\scriptsize \boldmath $g$}_{U}\in \mathcal{R}_{D}}P(\mbox{\boldmath $g$}_{U})P_{cD}(\mbox{\boldmath $g$}_{U}) + \sum_{\mbox{\scriptsize \boldmath $g$}_{U}\in \mathcal{R}_{Dc}}P(\mbox{\boldmath $g$}_{U})P_{mD}(\mbox{\boldmath $g$}_{U}). \label{GEP_D_in_Proof}
\end{eqnarray}

\subsection*{Bounding procedure}
To upper bound each term in (\ref{GEP_D_in_Proof}), let
\begin{equation*}
\mbox{\boldmath $X$}_{D}^{(N)} = \left[\mbox{\boldmath $X$}_{S}^{(N)}, \mbox{\boldmath $X$}_{D\backslash S}^{(N)}\right], \quad S \subset D.
\end{equation*}

For the \emph{incorrect decoding probability}, we apply the RCU bound, yielding
\begin{eqnarray}
		&& \!\!\!\!\!\!\! \sum_{\mbox{\scriptsize \boldmath $g$}_{U}\in \mathcal{R}_{D}}P(\mbox{\boldmath $g$}_{U}) P_{iS}(\mbox{\boldmath $g$}_{U}) \leq \nonumber \\
		&& \!\!\!\!\!\!\! \sum_{\mbox{\scriptsize \boldmath $g$}_{U}\in \mathcal{R}_{D}} P(\mbox{\boldmath $g$}_{U}) \mathbb{E}\Bigg[ \min \Bigg\{1,\sum_{\scriptsize\begin{array}{c}\tilde{\mbox{\boldmath $g$}}_{U}\in \mathcal{R}_{D} \\ \tilde{\mbox{\scriptsize \boldmath $g$}}_{S}=\mbox{\scriptsize \boldmath  $g$}_{S}\end{array}}\!\! 2^{N\sum_{[i,j]\in D \backslash S}r_{\tilde{g}_{ij}}}\times \nonumber\\
		&& \mathbb{P}\Big[\ell_{\tilde{\mbox{\scriptsize \boldmath $g$}}_{U}}\!\! \left(\mbox{\boldmath $X$}_{S}^{(N)}, \tilde{\mbox{\boldmath $X$}}_{D\backslash S}^{(N)},Y^{(NL+T_2)}\! \right) \! > \nonumber\\
		&& \ \ \ \ \ \ell_{\mbox{\scriptsize \boldmath $g$}_{U}}\left(\mbox{\boldmath $X$}_{D}^{(N)},Y^{(NL+T_2)}\right) \Big \vert \mbox{\boldmath $X$}_{D}^{(N)},Y^{(NL+T_2)}\Big]\Bigg\}\Bigg]\nonumber\\
		&& \!\!\!\!\!\!\! = \sum_{\mbox{\scriptsize \boldmath $g$}_{U}\in \mathcal{R}_{D}}\!\!\! \sum_{\scriptsize\begin{array}{c}\tilde{\mbox{\boldmath $g$}}_{U}\in \mathcal{R}_{D} \\ \tilde{\mbox{\scriptsize \boldmath $g$}}_{S}=\mbox{\scriptsize \boldmath  $g$}_{S}\end{array}} \!\!\!\!\!\!\! P(\mbox{\boldmath $g$}_{U}) \mathbb{E}\Bigg[\min \Bigg\{\frac{1}{|\mathcal{R}_{D_{\mbox{\scriptsize \boldmath $g$}_{S}}}|}, 2^{N\sum_{[i,j]\in D \backslash S}r_{\tilde{g}_{ij}}} \times \nonumber \\
		&& \mathbb{P}\Big[\ell_{\tilde{\mbox{\scriptsize \boldmath $g$}}_{U}}\left(\mbox{\boldmath $X$}_{S}^{(N)},\tilde{\mbox{\boldmath $X$}}_{D \backslash S}^{(N)}, Y^{(NL+T_2)}\right)> \nonumber \\
		&&\ \ \ \ \ell_{\mbox{\scriptsize \boldmath $g$}_{U}}\left(\mbox{\boldmath $X$}_{D}^{(N)}, Y^{(NL+T_2)}\right)\Big \vert \mbox{\boldmath $X$}_{D}^{(N)}, Y^{(NL+T_2)}\Big]\Bigg\}\Bigg], \label{IncDecThm5}
\end{eqnarray}
	with expectation taken with respect to
	\begin{eqnarray}
		&& \!\!\!\!\!\!\!\!\!\!\!\!\! P\left(\mbox{\boldmath $X$}_{D}^{(N)},Y^{(NL+T_2)}\right)=\nonumber \\
		&& \ \ \ \ \ \ P\left(Y^{(NL+T_2)}\middle \vert \mbox{\boldmath $X$}_{D}^{(N)}\right) P_{\mbox{\scriptsize \boldmath $g$}_{U}}\left(\mbox{\boldmath $X$}_{D}^{(N)}\right),\label{JointDistributionbarD}
	\end{eqnarray}
and $\mathcal{R}_{D_{\mbox{\scriptsize \boldmath  $g$}_{S}}} = \left\{\tilde{\mbox{\boldmath $g$}}_{S} \in \mathcal{R}_{D}: \tilde{\mbox{\boldmath $g$}}_{S} = \mbox{\boldmath $g$}_{S}\right\}$.

For the \emph{combined collision and miss-detection probability}, we apply the union bound, index swapping, and a change of measure \cite[Proposition~18.3]{IT2025}), leading to

	\begin{eqnarray}
		&& \!\!\!\! \sum_{\mbox{\scriptsize \boldmath $g$}_{U}\in \mathcal{R}_{D}}P(\mbox{\boldmath $g$}_{U})P_{cS}(\mbox{\boldmath $g$}_{U}) + \sum_{\mbox{\scriptsize \boldmath $g$}_{U}\in \widehat{\mathcal{R}}_{D}\cup \mathcal{R}_{Dc}}P(\mbox{\boldmath $g$}_{U})P_{mS}(\mbox{\boldmath $g$}_{U}) \leq \nonumber \\
		&& \sum_{\mbox{\scriptsize \boldmath $g$}_{U}\in \mathcal{R}_{D}} \sum_{\scriptsize\begin{array}{c}\tilde{\mbox{\boldmath $g$}}_{U}\in \widehat{\mathcal{R}}_{D}\cup \mathcal{R}_{Dc} \\ \tilde{\mbox{\scriptsize \boldmath $g$}}_{S}=\mbox{\scriptsize \boldmath  $g$}_{S}\end{array}}\!\!\!\!\!\! \Bigg ( P(\mbox{\boldmath $g$}_{U}) \times\nonumber\\
		&& \ \ \ \ \ \ \ \ \ \ \mathbb{P}\Big[\ell_{\mbox{\scriptsize \boldmath $g$}_{U}}\left(\left[\mbox{\boldmath $X$}_{S}^{(N)},\mbox{\boldmath $X$}_{D \backslash S}^{(N)}\right], Y^{(NL+T_2)} \right) \leq \gamma_S\Big] \nonumber \\
		&&  +  P(\tilde{\mbox{\boldmath $g$}}_{U}) 2^{N\sum_{[i,j]\in D\backslash S}r_{g_{ij}}} \times \nonumber\\
	    && \ \ \ \ \ \ \ \ \ \ \ \mathbb{P}\Big[\ell_{\mbox{\scriptsize \boldmath $g$}_{U}} \left(\left[\mbox{\boldmath $X$}_{S}^{(N)}, \tilde{\mbox{\boldmath $X$}}_{D \backslash S}^{(N)}\right],Y^{(NL+T_2)} \right) > \gamma_S \Big]\Bigg)\nonumber
	\end{eqnarray}
	\begin{eqnarray}
		&& = \sum_{\mbox{\scriptsize \boldmath $g$}_{U}\in \mathcal{R}_{D}} \sum_{\scriptsize\begin{array}{c}\tilde{\mbox{\boldmath $g$}}_{U}\in \widehat{\mathcal{R}}_{D}\cup \mathcal{R}_{Dc} \\ \tilde{\mbox{\scriptsize \boldmath $g$}}_{S}=\mbox{\scriptsize \boldmath  $g$}_{S}\end{array}}  \!\!\!\!\! \Bigg( P(\mbox{\boldmath  $g$}_{U})\times \nonumber \\
		&& \ \ \ \ \ \ \ \ \ \ \ \ \ \ \  \ \ \ \mathbb{E}\left[\mathbf{1}\left\{\ell_{\mbox{\scriptsize \boldmath $g$}_{U}}\left(\mbox{\boldmath $X$}_{D}^{(N)}, Y^{(NL+T_2)} \right) \leq  \gamma_S \right\}\right] \nonumber\\
		&& + P(\tilde{\mbox{\boldmath $g$}}_{U})  2^{N\sum_{[i,j]\in D\backslash S}r_{g_{ij}}}\times \nonumber\\
		&&\  \mathbb{E}\left[\mathbf{1} \left\{\ell_{\mbox{\scriptsize \boldmath $g$}_{U}} \left([\mbox{\boldmath $X$}_{S}^{(N)},\tilde{\mbox{\boldmath $X$}}_{D\backslash S}^{(N)}], Y^{(NL+T_2)} \right)  >  \gamma_S  \right\} \right] \Bigg) \label{**} \\		
		&& = \sum_{\mbox{\scriptsize \boldmath $g$}_{U}\in \mathcal{R}_{D}} \!\! \sum_{\scriptsize\begin{array}{c}\tilde{\mbox{\boldmath $g$}}_{U}\in \widehat{\mathcal{R}}_{D}\cup \mathcal{R}_{Dc} \\ \tilde{\mbox{\scriptsize \boldmath $g$}}_{S}=\mbox{\scriptsize \boldmath  $g$}_{S}\end{array}} \nonumber\\
		&& \mathbb{E} \Bigg[P(\mbox{\boldmath $g$}_{U}) \mathbf{1}\left\{\ell_{\mbox{\scriptsize \boldmath  $g$}_{U}} \left(\mbox{\boldmath $X$}_{D}^{(N)},Y^{(NL+T_2)} \right) \leq \gamma_S \right\} \label{CollMissThem5_1} \\
		&& \ \ \ \ \ \ \ +\frac{P_{\tilde{\mbox{\scriptsize \boldmath $g$}}_{U}}\left(Y^{(NL+T_2)}\middle \vert \mbox{\boldmath $X$}_{S}^{(N)}\right)P(\mbox{\boldmath $g$}_{U})P(\tilde{\mbox{\boldmath $g$}}_{U})}{2^{N\sum_{[i,j]\in S}r_{g_{ij}}}\ell_{\mbox{\scriptsize \boldmath $g$}_{U}}\left(\mbox{\boldmath $X$}_{D}^{(N)},Y^{(NL+T_2)}\right)}\times \nonumber \\
		&& \ \ \ \ \ \ \ \ \ \ \ \ \ \ \  \mathbf{1}\left\{\ell_{\mbox{\scriptsize \boldmath $g$}_{U}}\left(\mbox{\boldmath $X$}_{D}^{(N)}, Y^{(NL+T_2)} \right) > \gamma_S\right\}\Bigg],\label{CollMissThem5_3}
	\end{eqnarray}
where the triple $\big(\mbox{\boldmath $X$}_{D}^{(N)}, \tilde{\mbox{\boldmath $X$}}_{D\setminus S}^{(N)}, Y^{(NL+T_2)}\big)$ in (\ref{**}) follows distribution	
	\begin{eqnarray}
		 && P\left(\mbox{\boldmath $X$}_{D}^{(N)}, \tilde{\mbox{\boldmath $X$}}_{D\setminus S}^{(N)}, Y^{(NL+T_2)} \right) = \nonumber \\
		 && P_{\tilde{\mbox{\scriptsize \boldmath $g$}}_{U}}\left(Y^{NL+T_2}\middle \vert \mbox{\boldmath $X$}_{D}^{(N)}\right)P_{\tilde{\mbox{\scriptsize \boldmath $g$}}_{U}}\left(\mbox{\boldmath $X$}_{D}^{(N)}\right)  P_{\mbox{\scriptsize \boldmath $g$}_{U}}\left(\tilde{\mbox{\boldmath $X$}}_{D \setminus S}^{(N)}\right), \nonumber
	\end{eqnarray}
	and $P_{\tilde{\mbox{\scriptsize \boldmath $g$}}_{U}}\left(Y^{(NL+T_2)}\middle \vert \mbox{\boldmath $X$}_{S}^{(N)}\right)$ is the conditional probability of $Y^{(NL+T_2)}$ given the input sequence $\mbox{\boldmath $X$}_{S}^{(N)}$ and code index vector $\tilde{\mbox{\boldmath $g$}}_{U} = \left [\mbox{\boldmath $g$}_{S}, \tilde{\mbox{\boldmath $g$}}_{U \backslash S}\right]^{\top}$.
	
The optimal threshold is then derived from (\ref{CollMissThem5_1}) and (\ref{CollMissThem5_3}) as \cite[Theorem~18.6]{IT2025})

	\begin{eqnarray}
		 && \gamma_{S}^{*}\left(\tilde{\mbox{\boldmath $g$}}_{U}, \mbox{\boldmath $X$}_{S}^{(N)},Y^{(NL+T_2)}\right) =  \nonumber \\
		 && P(\tilde{\mbox{\boldmath $g$}}_{U}) 2^{-N\sum_{[i,j]\in S}r_{g_{ij}}}P_{\tilde{\mbox{\scriptsize \boldmath $g$}}_{U}}\left(Y^{(NL+T_2)}\middle \vert \mbox{\boldmath $X$}_{S}^{(N)}\right).\label{OptimalK_S}
	\end{eqnarray}
Substituting (\ref{OptimalK_S}) into (\ref{CollMissThem5_1}) and (\ref{CollMissThem5_3}) and applying Markov’s inequality to (\ref{CollMissThem5_1}), while noting that the indicator function is bounded above by one, yields the final bound:

	\begin{eqnarray}
		&& \!\!\! \sum_{\mbox{\scriptsize \boldmath $g$}_{U}\in \mathcal{R}_{D}}P(\mbox{\boldmath $g$}_{U})P_{cS}(\mbox{\boldmath $g$}_{U}) + \sum_{\mbox{\scriptsize \boldmath $g$}_{U}\in \widehat{\mathcal{R}}_{D}\cup \mathcal{R}_{Dc}}P(\mbox{\boldmath $g$}_{U})P_{mS}(\mbox{\boldmath $g$}_{U}) \leq \nonumber \\
		&&	\!\!\! 2\sum_{\mbox{\scriptsize \boldmath $g$}_{U}\in \mathcal{R}_{D}} \sum_{\tilde{\mbox{\scriptsize \boldmath $g$}}_{U}\in \widehat{\mathcal{R}}_{D}\cup \mathcal{R}_{Dc}}\!\!\!\! \min \Bigg\{P(\mbox{\boldmath $g$}_{U}), P(\tilde{\mbox{\boldmath $g$}}_{U}) 2^{N\sum_{[i,j]\in D\backslash S}r_{g_{ij}}} \times \nonumber \\
		&& \ \ \ \ \ \ \ \ \ \ \ \ \ \ \ \ \ \ \ \ \ \mathbb{E}\left[\frac{P_{\tilde{\mbox{\scriptsize  \boldmath $g$}}_{U}}\left(Y^{(NL+T_2)}\middle \vert \mbox{\boldmath $X$}_{S}^{(N)}\right)}{P\left(Y^{(NL+T_2)}\middle \vert \mbox{\boldmath $X$}_{D}^{(N)}, \mbox{\boldmath $g$}_{U \backslash D}\right)}\right]\Bigg\}.\label{CollMiss_Thm5_final} 	
	\end{eqnarray}

\subsection*{Final step}
Thus far, we have bounded all terms associated with $S \subset D$ in (\ref{GEP_D_in_Proof}). The terms for $S=D$ are bounded as the following

\begin{eqnarray}
		&& \sum_{\mbox{\scriptsize \boldmath $g$}_{U}\in \mathcal{R}_{D}}P(\mbox{\boldmath $g$}_{U})P_{cD}(\mbox{\boldmath $g$}_{U}) + \sum_{\mbox{\scriptsize \boldmath $g$}_{U}\in \mathcal{R}_{Dc}}P(\mbox{\boldmath $g$}_{U})P_{mD}(\mbox{\boldmath $g$}_{U})\leq \nonumber\\
		&& \sum_{\mbox{\scriptsize \boldmath $g$}_{U}\in \mathcal{R}_{D}}\!\!\!\! \sum_{\scriptsize\begin{array}{c}\tilde{\mbox{\boldmath $g$}}_{U}\in \mathcal{R}_{Dc} \\ \tilde{\mbox{\scriptsize \boldmath $g$}}_{D}=\mbox{\scriptsize \boldmath  $g$}_{D}\end{array}}\!\!\!\!\! P(\mbox{\boldmath  $g$}_{U})\mathbb{P}\left[\ell_{\mbox{\scriptsize \boldmath $g$}_{U}}\left(\mbox{\boldmath  $X$}_{D}^{(N)},Y^{(NL+T_2)}\right)\leq \gamma_D\right] \nonumber \\ 	        		&& +\!\!\!\! \sum_{\mbox{\scriptsize \boldmath $g$}_{U}\in \mathcal{R}_{Dc}}\!\!\!\! \sum_{\scriptsize\begin{array}{c}\tilde{\mbox{\boldmath $g$}}_{U}\in \mathcal{R}_{D} \\ \tilde{\mbox{\scriptsize \boldmath $g$}}_{D}=\mbox{\scriptsize \boldmath  $g$}_{D}\end{array}}\!\!\!\!\!\! P(\mbox{\boldmath $g$}_{U})\mathbb{P}\left[\ell_{\tilde{\mbox{\scriptsize \boldmath  $g$}}_{U}}\left(\mbox{\boldmath $X$}_{D}^{(N)},Y^{(NL+T_2)}\right)> \gamma_D\right]\nonumber\\
		&& = \!\!\!\!\! \sum_{\mbox{\scriptsize \boldmath $g$}_{U}\in \mathcal{R}_{D}} \!\!\!\! \sum_{\scriptsize\begin{array}{c}\tilde{\mbox{\boldmath $g$}}_{U}\in \mathcal{R}_{Dc} \\ \tilde{\mbox{\scriptsize \boldmath $g$}}_{D}=\mbox{\scriptsize \boldmath  $g$}_{D}\end{array}} \!\!\!\!\!\! \Bigg( P(\mbox{\boldmath $g$}_{U}) \mathbb{P}\left[\ell_{\mbox{\scriptsize \boldmath $g$}_{U}} \left(\mbox{\boldmath $X$}_{D}^{(N)},Y^{(NL+T_2)} \right) \leq \gamma_D\right] \nonumber\\
		&& \ \ \ \ \ \ \ \ \ + P(\tilde{\mbox{\boldmath  $g$}}_{U}) \mathbb{P}\left[\ell_{\mbox{\scriptsize \boldmath $g$}_{U}}\left(\mbox{\boldmath $X$}_{D}^{(N)},Y^{(NL+T_2)}\right)> \gamma_D\right]\Bigg)\label{Swapp2Thm5_D}\\
		&& = \!\!\!\!\! \sum_{\mbox{\scriptsize \boldmath $g$}_{U}\in \mathcal{R}_{D}}\!\!\!\! \sum_{\scriptsize\begin{array}{c}\tilde{\mbox{\boldmath $g$}}_{U}\in \mathcal{R}_{Dc} \\ \tilde{\mbox{\scriptsize \boldmath $g$}}_{D}=\mbox{\scriptsize \boldmath  $g$}_{D}\end{array}} \!\!\!\!\!\! \mathbb{E}\Bigg[P(\mbox{\boldmath $g$}_{U})\mathbf{1}\left \{ \ell_{\mbox{\scriptsize \boldmath $g$}_{U}}\left(\mbox{\boldmath $X$}_{D}^{(N)},Y^{(NL+T_2)}\right)\! \leq\!  \gamma_D\right \}\label{ChangeMeasure1} \\
		&&  + \frac{P(\tilde{\mbox{\boldmath $g$}}_{U})P(\mbox{\boldmath $g$}_{U})P_{\tilde{\mbox{\scriptsize \boldmath $g$}}_{U}}\left(Y^{(NL+T_2)}\middle \vert \mbox{\boldmath $X$}_{D}^{(N)}\right)}{2^{N\sum_{[i,j]\in D}r_{g_{ij}}}\ell_{\mbox{\scriptsize \boldmath $g$}_{U}}\left(\mbox{\boldmath $X$}_{D}^{(N)},Y^{(NL+T_2)}\right)}\times \nonumber \\
		&& \ \ \ \ \ \ \ \ \ \ \ \ \ \ \ \ \ \ \ \ \ \mathbf{1}\left\{\ell_{\mbox{\scriptsize \boldmath $g$}_{U}}\left(\mbox{\boldmath $X$}_{D}^{(N)},Y^{(NL+T_2)}\right)> \gamma_D\right\}\Bigg]\label{ChangeMeasure2}\\
		&& \leq 2 \sum_{\mbox{\scriptsize \boldmath $g$}_{U}\in \mathcal{R}_{D}} \sum_{\scriptsize\begin{array}{c}\tilde{\mbox{\boldmath $g$}}_{U}\in \mathcal{R}_{Dc} \\ \tilde{\mbox{\scriptsize \boldmath $g$}}_{D}=\mbox{\scriptsize \boldmath  $g$}_{D}\end{array}} \min \Bigg\{P(\mbox{\boldmath $g$}_{U}), \nonumber \\
		&& \ \ \ \ \ \ \ \ \ \ \ \ \ P(\tilde{\mbox{\boldmath $g$}}_{U}) \mathbb{E}\Bigg[\frac{P_{\tilde{\mbox{\scriptsize \boldmath $g$}}_{U}}\left(Y^{(NL+T_2)}\middle \vert \mbox{\boldmath $X$}_{D}^{(N)}\right)}{P\left(Y^{(NL+T_2)}\middle \vert \mbox{\boldmath $X$}_{D}^{(N)},\mbox{\boldmath $g$}_{U \backslash D}\right)}\Bigg]\Bigg\},\label{FinalBound_D_Thm5}
	\end{eqnarray}
	with optimal threshold
	\begin{eqnarray}
		&& \!\!\!\! \gamma_{D}^{*}\left(\tilde{\mbox{\boldmath $g$}}_{U}, \mbox{\boldmath $X$}_{D}^{(N)},Y^{(NL+T_2)}\right) = \nonumber \\
		&& \!\!\!\! P(\tilde{\mbox{\boldmath $g$}}_{U}) 2^{-N\sum_{[i,j]\in D}r_{g_{ij}}}P_{\tilde{\mbox{\scriptsize \boldmath $g$}}_{U}}\left(Y^{(NL+T_2)}\middle \vert \mbox{\boldmath $X$}_{D}^{(N)}\right),
	\end{eqnarray}
where $P_{\tilde{\mbox{\scriptsize \boldmath $g$}}_{U}}\left(Y^{(NL+T_2)}\middle \vert \mbox{\boldmath $X$}_{D}^{(N)}\right)$ is the conditional probability of $Y^{(NL+T_2)}$ given $\mbox{\boldmath $X$}_{D}^{(N)}$ and the vector of code indices $\tilde{\mbox{\boldmath $g$}}_{U} = \left[\mbox{\boldmath $g$}_{D}, \tilde{\mbox{\boldmath $g$}}_{U \backslash D}\right]^{\top}$.

Substituting (\ref{IncDecThm5}), (\ref{CollMiss_Thm5_final}), and (\ref{FinalBound_D_Thm5}) into (\ref{GEP_D_in_Proof}) completes the proof.
\end{proof}

\end{document}